\documentclass[smallcondensed]{svjour3}                     

\smartqed  
\usepackage{graphicx}

\usepackage{multirow}
\usepackage{cite}      
\usepackage{subfigure} 
\usepackage{amsmath}   
\usepackage{amssymb}
\usepackage{array}
\usepackage{amsfonts}
\usepackage{algorithm}
\usepackage{algorithmic}

\newcommand{\comment}[1]{}

\newtheorem{thm}{Theorem}

\newtheorem{lem}{Lemma}

\begin{document}

\title{In-Network Outlier Detection in Wireless Sensor Networks}

\author{Joel W. Branch \and 
Chris Giannella \and 
Boleslaw Szymanski \and 
Ran Wolff \and 
Hillol Kargupta}

\institute{J. W. Branch \at 
Network Management Research Department, IBM T.J. Watson Research Center, Hawthorne, NY 10532, USA, \email{branchj@us.ibm.com}
\and
C. Giannella \at
The MITRE Corporation, 300 Sentinel Dr Ste 600, Annapolis Junction, MD 20701 USA, \email{cgiannel@acm.org}
\and
B. Szymanski \at 
Department of Computer Science, Rensselaer Polytechnic Institute, Troy, NY 12180, USA, \email{szymansk@cs.rpi.edu}
\and
R. Wolff \at
Department of Management Information Systems, University of Haifa, Haifa, Israel, \email{rwolff@mis.haifa.ac.il}
\and 
H. Kargupta \at 
Department of Computer Science and Electrical Engineering, University of Maryland, Baltimore County, Baltimore, MD 21250, USA, \email{hillol@cs.umbc.edu} \\ 
Also affiliated with AGNIK, LLC, USA
}

\maketitle

\begin{abstract}

To address the problem of unsupervised outlier detection in wireless sensor 
networks, we develop an approach that (1) is flexible with respect to the outlier 
definition, (2) computes the result in-network to reduce both bandwidth and energy usage,
(3) only uses single hop communication thus permitting very simple node failure detection
and message reliability assurance mechanisms (e.g., carrier-sense), and (4) seamlessly 
accommodates dynamic updates to data.  We examine performance using 
simulation with real sensor data streams. Our results demonstrate that our approach is accurate and 
imposes a reasonable communication load and level of
power consumption.


\end{abstract}


\section{Introduction}

Outlier detection, an essential step preceding most any data 
analysis routine,  
is used either to suppress or amplify outliers.
The first usage (also known as data cleansing) improves robustness of data
analysis. The second usage helps in searching for rare patterns in such domains 
as fraud analysis,
intrusion detection, and web purchase analysis (among others).

Several factors make wireless sensor networks (WSNs) especially prone to 
outliers. First, they collect their data from the real world using imperfect
sensing devices. Second, they are battery powered and thus their performance
tends to deteriorate as power dwindles. Third, since these networks
may include a large number of sensors, the chance of error accumulates.
Finally, in their usage for security and military purposes, sensors are
especially prone to manipulation by adversaries. Hence, it is clear that
outlier detection should be an inseparable part of any data processing routine
that takes place in WSNs.

Simply put, outliers are events with extremely small probabilities of occurrence.
Since the actual generating distribution of the data is usually unknown, 
direct computation of probabilities is difficult. Hence, outlier detection 
methods are, by and large, heuristics. Because
the problem is fundamental, a huge variety of outlier detection methods have been developed. In this paper we focus on non-parametric, unsupervised methods. 
A simplistic implementation of these methods would require centralization of the data. Such centralization
is hard and costly in WSNs as it demands high bandwidth and requires reliable message transmission over multiple hops, 
which is both costly and difficult to implement.

We developed a technique for the computation of outliers in WSNs. 
This technique (1) is flexible with respect to the outlier
definition, (2) computes the result in-network to reduce both bandwidth and energy 
usage \cite{GuptaKumar},
(3) only uses single hop communication thus permitting very simple node failure detection
and message reliability assurance mechanisms (e.g., carrier-sense), and (4) seamlessly 
accommodate
dynamic updates to data. In addition to these essential features, the algorithm  
presented here also has two highly desirable properties: it is  generic -- suitable for many outliers 
detection heuristics and its communication load is proportional to the outcome ({\em i.e.} the number of outliers 
reported).

We exemplify the benefits of our algorithm by implementing it using two
different outlier detection heuristics and simulating 53 sensors using the
SENSE sensor network simulator \cite{sense:2004} with real sensor data streams.
Our results show that the algorithm converges to an accurate result with reasonable communication load and power consumption. In most tested cases, our algorithm's 
performance bests that of a centralized approach.

\section{Motivating Application}

The potential importance of efficient outlier detection in wireless
sensor networks is best understood in the context of popular applications
of those systems. Consider, for instance, the acoustic source localization
problem. In this problem, a set of synchronized sensors all register the
arrival of a specific sound at a certain time. Given the distance of
two sensors from one another and the time difference of arrival (TDOA) of the
sound, the potential locations of the source vis-a-vis the two sensors can be 
deduced. Given data from several sensors, the possible relative locations (each 
a hyperbola in the plane) can be intersected, and the location of the source 
can be
pinpointed (see, for example \cite{Ajdler04sourceLocalization,CounterSniper} and Fig. \ref{fig:tdoa-sl}).

\begin{figure}[t]

\begin{center}

\includegraphics[width=4in]{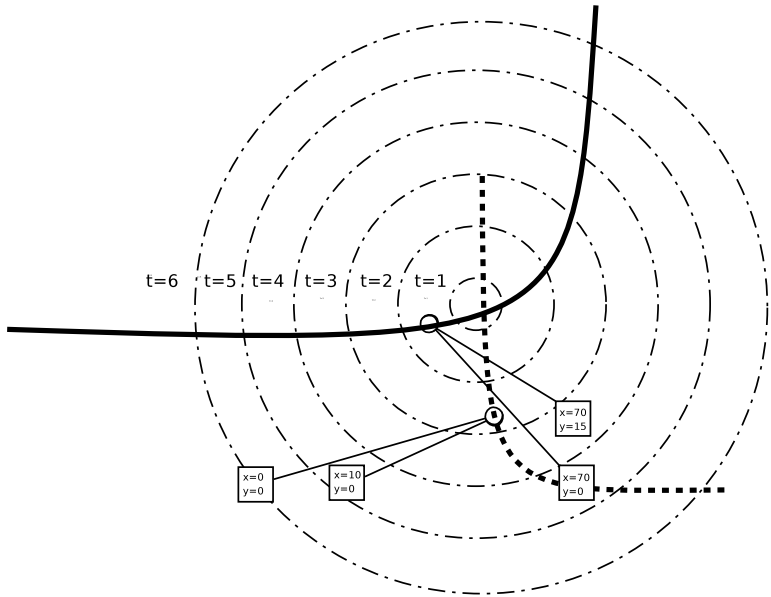}

\caption{The expansion of a sound over time and the possible source location as computed by two different pairs
of sensors according to the time difference of arrival. The origin of the sound lays in the intersection
of the two hyperbola.}
\label{fig:tdoa-sl}
\end{center}

\end{figure}

While the theoretic framework of TDOA based source location is simple and clear,
the problem becomes much more complex in reality. Firstly, the real terrain in
which the problem occurs is rarely flat, or an unobstructed three-dimensional space. 
Secondly, echos and multiple concurrent
sounds may add many possible hyperbolas from which the relevant ones need be selected.
Last, and perhaps most importantly, the method is sensitive to erroneous initializations
in terms of sensor synchronization and positioning, as well as to the possible degradation
of these factors when sensors' power dwindles. All these factors amount to a multiplicity
of possible hyperbolas, only few of which intersect at the correct location of the source.

In fact the similar principal of localization applies in a broader setting,
in which imprecise detection by a single sensor, regardless
of its modality (e.g., acoustic, seismic, visual, electromagentic, etc.) is
applied in so called {\it binary sensing} object location~\cite{WBS09} in which 
a group of neighboring nodes cooperate to narrow the object location. In fact,
a detection, true of false, of object presence in a sensing range of sensor
will trigger a tracking algorithm or even the entire tracking service~\cite{CJ-SOA09}.
Hence, to avoid the costs associated with unnecessary execution of the tracking 
algorithm or service, wrong data (whatever the cause)
must be detected and removed. There are ample methods in which such detection and removal 
can be carried out ({\em e.g.} Maximum Likelihood \cite{ML-WSN}). 
However, they all rely on centralization of all of the data for processing.
Such centralization would likely be unacceptably costly in wireless sensor networks
for two main reasons.  First, because a huge portion of the energy of a sensor would 
be spent on relaying data of other sensors. Second, because naive centralization would 
make no use of old data when new data arrives, even if the data changes  only slightly. 
For instance, if a certain sensor 
produces unwanted signals (say, due to a local noise source), that sensor \emph{and 
every sensor relaying its data to the center} would constantly waste energy on 
centralization of the data even though it might clearly be undesired.

It is therefore of high importance to be able to perform data cleansing in the network
concurrent to any decision protocol. With the method suggested in this paper, 
sensors can constantly and efficiently prune away data which seems false. Only then, 
and only if the remaining data seems to require further analysis, would the more 
complex and costly procedure for source localization be executed. In this way, much 
energy can be saved and system lifetime can be extended. 

This paper presents an efficient algorithm  for  in-network outlier detection. 
The algorithm 
is generic, permitting several definitions of an outlier. The experimentation
and simulation results are presented for this aglorithm and not for the entire
motivating example because source localization and tracking using wireless sensor 
networks is well understood (e.g., see~\cite{beck2008exact}).


\section{Related work}

\subsection{Outlier detection}

Outlier detection is a long studied problem in data analysis, 
we provide only a brief sampling of the field.
Hodge and Austin \cite{HA:2004} present a survey focusing on
outlier detection methodologies based on machine learning and data mining. These include
distance and density-based unsupervised methods, feed-forward 
neural networks and decision tree-based supervised methods, and 
auto-associative neural network and Hopfield network-based methods. 
Barnett and Lewis \cite{BL:1994} provide a
survey of outlier detection methodologies in the statistics community.

Our algorithm is flexible in that it accommodates a whole class of unsupervised
outlier detection techniques such as (1) distance to $k^{th}$ nearest neighbor
\cite{RRS:2000}, \cite{BS:2003}, (2) average distance to the $k$ nearest neighbors \cite{AP:2002}, \cite{BS:2003}, and
(3) the inverse of the number of neighbors within a distance $\alpha$ \cite{KN:1998}
(see Section \ref{sec:prelim} for details).

\subsection{Wireless sensor networks}

WSNs combine the capability to sense, compute, and coordinate their activities with
the ability to communicate results to the outside world. They are 
revolutionizing data collection in all kinds of environments. At the same
time, the design and deployment of these networks creates unique research and
engineering challenges due to their expected large size (up to thousands of
sensor nodes), their often random and hazardous deployment, obstacles to their
communication, their limited power supply, and their high failure rate. 

The software for WSNs needs to be aware of their limitations and 
features. The most important among these are limited power, high
communication cost, and limited direct communication range. In 
\cite{mobicom99}, Estrin {\em et al.} introduce scalable coordination
as an important component of the needed software. A survey of the
state-of-the-art in WSNs is given in 
in \cite{computer-networks02}. Another survey~\cite{comm-magazine02} focuses 
on challenges arising from specific applications such as military, 
health care, ecology, and security.

Energy-efficiency, a cardinal WSN requirement, is often achieved by minimizing communication using
topology-control algorithms that dictate the active/sleep cycles of sensor 
nodes. Examples include Geographic Adaptive Fidelity (GAF) 
\cite{mobicom01}, ASCENT \cite{ascent},
STEM \cite{stem},
and ESCORT \cite{icn05}.
While the focus of this paper is on WSN outlier detection,
the challenge is the same as in the above mentioned works. Hence, while we do not propose
a topology-control algorithm, we aim to design an
energy-efficient algorithm by minimizing the required communication overhead.

Other research efforts have also addressed the issue of developing a framework
for distributed outlier detection in WSNs.

The framework of Zhuang {\em et al.} \cite{ZCWL:2007} use a weighted moving average approach 
to smooth noise from the data stream arriving at each sensor.  In addition to temporal information 
(past data values), sensors also use data from neighboring sensors (spatial smoothing) to reduce the
rate at which data values are propagated to the sink.  When an observed data value remains within the
established spatio-temporal trend, it is not propagated.  Their approach differs from ours in that theirs
does not seek to detect outliers. 

The framework of Sheng {\em et al.} \cite{SLMJ:2007} allows the discovery of k-nearest-neighbor
based outliers: points whose distance to their k-nn exceeds a fixed threshold or the top n points
with respect to the distance to their k-nns.  Each sensor maintain a histogram-type summary of pertinent 
information over a sliding window of its data points.  This summary is propagate to a sink node.  The sink 
node collects the summaries and queries the 
network for any additional information needed to correctly determine the outliers over the whole network.
The use of summaries allows their approach to use less communication than a naive, centralized approach.
Their approach differs from ours in several ways.  First, they only detect outliers over one dimensional data.
Indeed, extending their approach to more dimensions is complicated by the fact that compact, multi-dimensional
histograms are difficult to build.  Second, they only consider the two k-nn based outlier definitions described
above.  While our approach encompasses these and more.  Thirdly, their approach only applies in settings where
spatial proximity is unimportant (data from all sensors, near and far, is used in determining outliers).  We have
developed an approach that considers spatial proximity ("semi-local" outlier detection) as well as one that does not.

The framework of Subramaniam {\em et al.} \cite{SPPKG:2006} requires the
sensors to maintain a tree communication topology and computes outliers
based on an estimate of the underlying probability distribution from which the data 
arises.  Such an estimate is computing by each sensor maintaining a random sample of 
its data observations.  Our approach differs in at least four ways.  First, ours
does not make any assumptions about the communication
topology ({\em e.g.} it is a tree), save that it is connected.  Second, ours computes
outliers with respect to all of the data observations at each sensor, not a sample.
Third, ours can smoothly take into account spatial proximity among the
sensors (``semi-local" outliers) while Subramaniam does not focus on this task.  
Fourth, our approach is designed to smoothly adjust to changes in the underlying 
network topology while Subramaniam's requires that the underlying communication tree
be reestablished by other means before their algorithm can resume operation.

The framework of Janakiram {\em et al.} \cite{JRK:2006} is based on a Bayesian
Belief Network (BBN) that has been constructed over the WSN (and distributed to each sensor).
Using this, each sensor can estimate the likelihood of an observed tuple and, therefore, detect
outliers.  However, Janakiram does not discuss the problem of updating the BBN given
network/data change.  It is not clear to what extent the BBN construction phase can by
carried out in-network.  Our approach differs in that it is in-network and designed to
smoothly adjust to changes in data/network.   

The framework of Zhuang and Chen \cite{ZC:2006} uses a wavelet based technique for
correcting large isolated spikes from single sensor data streams. A dynamic time warping (DTW)distance-based 
technique is also used to identify more steady intervals of erroneous sensor data by comparing the data streams of 
spatially close sensors assumed to produce similar
data streams. To reduce energy consumption, anomalous data streams are not transmitted
to the base station. Our method is similar in that it is in-network. However, Zhuang
and Chen's use of DTW is tightly integrated with a minimum hop count routing algorithm,
which makes the approach more restrictive than ours.

Rajasegarar {\em et al.} \cite{RLPB:2006} describe an approach that is based on
distributed non-parametric anomaly detection and requires sensors to maintain
a tree communication network topology. Here each sensor clusters its sampled
measurements using a fixed-width clustering algorithm, then extracts statistics
of the clusters (i.e., the centroid and number of contained data vectors) and
then sends them its parent node. The parent uses its children's' cluster 
statistics to form
a merged cluster and then transmits that cluster to its parent. This process continues
recursively until the base station receives all clusters, after which it will
perform anomaly detection to identify all outliers. While this approach supports
energy-efficiency by distributing the clustering operation throughout the network,
anomaly detection is only performed at the base station. Our approach differs in
that it distributes the anomaly detection process itself throughout the network,
quickly enabling nodes to identify outliers and autonomously make further data
processing decisions. Also, our approach does not rely on the use and maintenance
of a routing tree and hence, is able to smoothly adjusts to changes in the underlying 
network topology.

Adam {\em et al.} \cite{AJA:2004} address the issue of accounting for spatially neighboring
peers when detecting outliers in sensor networks.  However, they assume the sensor datasets
are centralized and the outlier processing is carried out there.  They
do not consider the problem of carrying out the outlier detection {\em in-network}
as we do. 

Palpanas {\em et al.} \cite{PPKG:2003} propose a technique for distributed deviation
detection using a network hierarchy of low and high capacity sensors that are differentiated
with respect to processing power and communication range. Here, low capacity sensors
aim to detect local outliers while high capacity sensors detect more spatially dispersed
outliers using an aggregation of low capacity sensors' data. Kernel density estimators
are used to model the distribution of data values reported by sensors and distance-based
detection techniques are used for identifying outliers. The authors present no formal
evaluation of the proposed technique. Our approach differs in that it
does not rely on a hierarchy of device capabilities.

The framework of Radivojac {\em et al.} \cite{RKSO:2003} addresses the process of
sensors learning data distributions from class-imbalanced data. Here, sensors send data
points to a central base station which is tasked with generating a classification model
from class-imbalanced data (i.e., having an abundant number of negative samples and a
small amount of positives). The model is generated using a neural network classifier,
after which the base station distributes the model to the sensors for detection purposes.
This process repeats throughout the lifetime of the network. A Bayesian classifier is also
employed to extend the lifetime of the network by minimizing
the total cost of detection and classification (e.g., costs of transmitting false-positives
and false-negatives). Again, our framework differs in that it operates in-network as
opposed to a centralized manner.

Our work in this paper is an extension of our preliminary work appearing in conference 
proceedings \cite{BSWGK:2006}.  We have extended our preliminary work by 
providing complete correctness proofs for the global outlier detection algorithm.  
And, we have improved the experimental analysis of the global algorithm.  We have also
added the localized outlier detection algorithm and experimental analysis of it. 

\subsection{Distributed data mining}

Distributed Data Mining (DDM) has recently emerged as an important area of research.
DDM is concerned with analysis of data in distributed environments,
while paying careful attention to issues related to computation, communication, storage, and human-computer
interaction.  Detailed surveys of
Distributed Data Mining algorithms and techniques have been presented in
\cite{Kargupta_97}, \cite{Kargupta_99}, \cite{Kargupta_04}.
Some of the common data-analysis tasks include association rule mining, clustering,
classification, kernel density estimation and so on.

Recently, researchers have started to consider data analysis and data mining in 
large-scale dynamic networks with the goal of developing techniques that are 
highly asynchronous, scalable, and robust to network changes. Efficient data
analysis algorithms often rely on efficient primitives, so
researchers have developed several different approaches to computing basic
operations ({\em e.g.} average, sum, max, or random sampling) on dynamic networks.
Mehyar {\em et al.} \cite{MSPLM:2007} develop an asynchronous, deterministic technique for computing
an average over a large, dynamic network.  Kempe {\em et al.} \cite{kempe03gossip} and Boyd {\em et al.} \cite{BGPS:2005}
investigate gossip based randomized algorithms. Jelasity and Eiben \cite{KJE:2003}
develop the ``newscast model''.  Bawa {\em et al.}
\cite{BGMM:2007} have developed an approach in which similar primitives are
evaluated to within an error margin. Wolff {\em et al.} \cite{Majority-Rule-SMC-B}
develop a local algorithm for majority voting. Datta and Kargupta 
\cite{DK:2007} develop a technique for uniformly sampling data distributed over
a large-scale peer-to-peer network.  Wolff {\em et al.} \cite{WBK:2009}, Sharfman {\em et al.} \cite{SSK:2007},
and Bhaduri {\em et al.} \cite{BK:2008} develop
techniques for threshold monitoring over a large, distributed set of data streams.  Finally, some work has gone into more
complex data mining tasks: association rule mining \cite{Majority-Rule-SMC-B},
facility location \cite{KSW:2007}, decision tree induction \cite{BWGK:2008}, 
classification through meta-learning \cite{LXLS:2007} 
(all four based on local majority voting), genetic
algorithms \cite{CDS:2003}, k-means clustering \cite{DGK:2006} 
\cite{WKK:2006}, web user community formation \cite{DBLK:2008}, hidden 
variable distribution estimation in a wireless sensor network \cite{MK:2008}, 
outlier detection in distributed data streams
\cite{OGP:2006} \cite{SHYZJ:2007}.  The last two papers address a related problem as 
we do: outlier detection over multiple distributed data streams.  However, their work is not designed for
a WSN.  For example, they rely on frequent whole-network broadcasts (Otey) or
information centralization at a leader node (Su) --
arguably reasonable approaches in a wired network, but very costly in a WSN. 
Finally, an overview of the problem of carrying out data mining on data distributed over
a dynamic peer-to-peer network is given in \cite{DBGWK:2006}. 


\section{Preliminaries}
\label{sec:prelim} 

\subsection{Outlier Detection Defined}

Let $\mathbb{D}$ be a data space.  We adopt a commonly used approach in the data mining/machine learning literature such that 
outliers are defined by specifying ranking function, $R$.  This function maps $x \in \mathbb{D}$ and finite $D \subseteq \mathbb{D}$ 
to a non-negative real number $R(x,D)$ indicating the degree to which $x$ can be regarded as an outlier with respect to a 
dataset $D$.  Some common examples of $R$ include (among others): the distance to the $k^{th}$ nearest neighbor 
(\cite{RRS:2000}, \cite{BS:2003}); the average distance to the $k$ nearest neighbors (\cite{AP:2002}, \cite{BS:2003}); and LOF 
(\cite{breunig00lof}).  We assume that a fixed total linear order, $\prec$, on $\mathbb{D}$ is 
used as a tie-breaking mechanism to
ensure that $R(.,Q)$ creates a total linear ordering on $\mathbb{D}$ for any finite $Q \subseteq \mathbb{D}$.  This is equivalent,
for our purposes, to assuming, without loss of generality, that $R(.,Q)$ is one-to-one.   

$R$ is assumed to satisfy the following two axioms. Given $x \in \mathbb{D}$, for all finite
$Q_1 \subseteq Q_2 \subseteq \mathbb{D}$: {\bf anti-monotonicty,} $R(x,Q_1) \geq R(x,Q_2)$; {\bf smoothness,}
if $R(x,Q_1) > R(x,Q_2)$, then there exists $z \in Q_2 \setminus Q_1$, such that $R(x,Q_1) > R(x,Q_1 \cup \{z\})$. 
The anti-monotonicty axiom is similar to the {\em Apriori rule} in frequent itemset 
mining \cite{AMS:1996}. The smoothness axiom, intuitively, states that $R$ changes
gradually. As more points are added to $Q_1$, the rating function changes gradually
to $R(x,Q_2)$.  Of the examples in the previous paragraph, all but LOF satisfies these assumptions, assuming,
as we do, the use of a tie-breaking mechanism as described in the previous paragraph.  
 
Given $n$ a user-defined parameter and a finite dataset $D \subseteq \mathbb{D}$, the
outliers of $D$ are denoted $O_n(D)$ and are defined to be the top $n$ points
in $D$ with respect to $R(.,D)$ (if $|D| < n$, then $O_n(D)$ is defined to be $D$). 

\subsection{Distributed System Set-up}

The distributed system architecture we assume consists of a collection of sensors, $p_i$, each holding 
a finite dataset $D_{i} \subseteq \mathbb{D}$. 
Sensors communicate by exchanging messages to their
immediate neighbors as defined by an undirected graph.  
We assume that messages are reliable, {\em i.e.} a message sender can assume that if a message is not recieved, then
the sender will be informed); and each sensor 
$p_{i}$ can accurately maintain the list
of its immediate neighbors, $\Gamma_{i}$, in the graph.
Our algorithms
work as long as there exists a path, possibly unknown, from each sensor
to every other sensor.  Note that, message reliability is difficult to fully
maintain in a WSN -- some message dropping is expected.
While our algorithm assumes no message dropping, modest violation of this 
assumption in our experiments did not effect accuracy significantly.


\section{Global Distributed Outlier Detection Algorithm}

In this section, we describe a distributed algorithm by which sensors, each
assumed to know $R$ and $n$, compute $O_n(D)$ where $D$ $=$ $\bigcup_{i}D_{i}$
({\em global} outlier detection).  In a wireless sensor network, it can be desirable 
for sensors to find outliers only with respect to the data contained in nearby sensors, rather than
the entire network ({\em semi-global} outlier detection).   
In the next section \ref{sec:localized}, we describe how to modify the global
outlier detection algorithm to act in a semi-global manner. 

At any point in time, $p_i$ keeps track of the data points it has sent to or received from
its neighbor $p_j$ at some past time.  Let $D^i_{i,j}$ denote the set of points sent from $p_i$ to 
$p_j$, and, $D^i_{j,i}$ deonte the set of points sent 
from $p_j$ to $p_i$.  Importantly, $(D^i_{i,j}$ $\cup$ 
$D^i_{j,i})$ denotes the data points that $p_i$ can be sure
are commonly held with $p_j$ (there may be more).  
Let $P_i$ denote 
$D_i \bigcup_{j \in \Gamma_i}D^i_{j,i}$, the set of points $p_i$ is holding
at the current time.\footnote{Note the distinction between $D_i$ and $P_i$.  In words, $D_i$ is the set of points 
that {\em originated} at sensor $p_i$, while $P_i$ is the set of all points that $p_i$ is holding including $D_i$ {\em and} 
those originating at other sensors but propogated to sensor $p_i$ through messaging passing.}  $p_i$ uses $P_i$ to compute 
an estimate of the overall correct answer, $O_{n}(D)$.  This estiamte, henceforth called
$p_i's$ {\em estimate}, is $O_n(P_i)$, the set of outliers based on all the information 
availible to $p_i$ at the current time.  

The algorithm does not assume any special sensors.  Each sensor, $p_i$, asynchronously waits for an 
{\em event} to occur: (i) the algorithm is
initialized, (ii) $D_i$ changes, (iii) a message is received from a neighbor, or (iv) a link goes up/down
causing $p_i's$ immediate neighborhood to change (however, algorithm correctness requires that we assume
the network remain connected). Note that,
events for $p_i$ are entirely local and can be detected without the aid of any other sensors beyond
the immediate neighborhood.  
Once $p_i$ detects an event, it will decide which of
the points it is currently holding ($P_i$), if sent, could cause its neighbor, $p_j$, to 
change its estimate.  $p_i$ then sends these points and adds them to
$D^i_{i,j}$ ($p_i$ carries out this process separately for
all of its neighbors).  

Gradually, the points held by each sensor enlarges until enough overlap 
is obtained so that each sensor's estimate is
the correct answer, $O_{n}(D)$.  This will be gauranteed to occur once
each sensor, individually, decides that none of the points it is currently
holding need be sent to its neighbors.  At this point, the algorithm is
terminated.  To see how all of this works, consider an example.

\subsection{Example}

Let $R$ be the distance to the nearest neighbor and given $x \in \mathbb{D}$ and
finite $P \subseteq \mathbb{D}$, $N(x,P)$ denotes the nearest neighbor of $x$ among points in $P$.
Given finite $Q \subseteq \mathbb{D}$, $N(Q,P)$ denotes $\bigcup_{x \in Q}N(x,P)$.  Let $n=1$
and consider a network
of two sensors, $p_i$ and $p_j$, each initially holding the 
following one-dimentional datasets.  The correct 
answer the algorithm will compute is $O_{n}(D)$ $=$ $\{0.5\}$. 

\begin{itemize}
\item $D_i = \{0.5,3,6,10,11, \ldots, a\}$.
\item $D_j = \{4,5,7,8,9,a+1,a+2, \ldots,a+b\}$.  
\item $D = D_i \cup D_j = \{0.5,3,4,5,6,7,8,9,10,11, \ldots, a, a+1, \ldots, a+b\}.$
\end{itemize}

\noindent Initially, $P_i = D_i$, $P_j = D_j$, and $D^{i}_{i,j}$ 
$=$ $D^{i}_{j,i}$ $=$ $D^{j}_{i,j}$ 
$=$ $D^{j}_{j,i}$ $=$ $\emptyset$.  For simplicity, we will 
describe the 
algorithm in synchronous fashion starting with $p_i$.  But, the ideas 
extend nicely to asynchronous operation. 

\begin{enumerate}
\item $p_i$ computes its estimate as $O_{n}(P_i) = \{6\}$ and then must 
compute the set of its data points that might 
cause $p_j$ to change its estimate if sent.  We call these the {\em
sufficient points} from $P_i$ for sensor $p_j$.  Formally, we define a
set $Z_j \subseteq P_i$ to be sufficient for $p_j$ if 

\begin{equation}
\label{neccessary}
\left[O_{n}(P_i) \cup N(O_{n}(P_i),P_i)\right] \cup \left[N(O_{n}(D^{i}_{i,j}\cup D^{i}_{j,i} \cup Z_j),P_i)\right] \subseteq Z_j.
\end{equation}

\noindent The rationale for the first part is simple.  $O_{n}(P_i)$ is 
necessary for $p_j$, because if $p_i$ is right in its
estimate, then $p_j$ ought to know about it.  Moreover, $p_i$ must
also send $N(O_{n}(P_i),P_i)$, because this allows $p_j$
to determine if any of its points can cause the ranking of the points in   
$O_{n}(P_i)$ to change.  

The rationale for the second part is somewhat
more complicated.  In brief, if $p_i$ were to send $Z \subseteq P_i$ to 
$p_j$, then $p_i$ would also need to send 
$N(O_{n}(D^i_{i,j}\cup D^i_{j,i} \cup Z),P_i)$.
To avoid resending, $p_i$ requires that 
$N(O_{n}(D^i_{i,j}\cup D^i_{j,i} \cup Z),P_i)$ is
contained in $Z$.
To understand the reasoning for all of this, consider that
$(D^i_{i,j}\cup D^i_{j,i} \cup Z)$ is the total set of 
points $p_i$ knows $p_j$ has if $Z$ were to send $p_j$.  Thus,
$O_{n}(D^i_{i,j}\cup D^i_{j,i} \cup Z)$ is $p_i's$ best approximation 
to $p_j's$ estimate if $Z$ were to send $p_j$.
Hence, $p_i$ computes its nearest neighbors to these, because, if
$p_i$ is right in its approximation, it must ensure that $p_j$ 
have these neighbors since they could cause $p_j$ to change its estimate.

\vspace{0.25cm}

Getting back to our example, observe that $O_{n}(P_i)$ $\cup$
$N(O_{n}(P_i),P_i)$ $=$ $\{3,6\}$.  Moreover, 
$N(O_{n}(D^i_{i,j}\cup D^i_{j,i} \cup \{3,6\}),P_i)$
$=$ $N(O_{n}(\{3,6\}),P_i)$ $=$ $\{3\}$ $\subseteq$ $\{3,6\}$.  Hence,
$Z_j$ $=$ $\{3,6\}$ as this satisfies (\ref{neccessary}) 
above.  So, $p_i$ sends 
$\{3,6\} \setminus (D^i_{i,j} \cup D^i_{j,i})$ $=$ 
$\{3,6\}$ and updates $D^i_{i,j}$ to $\{3,6\}$.

\item $p_j$ will receive these points and updates $D^j_{i,j}$ to $\{3,6\}$ (currently 
$D^j_{j,i} = \emptyset$).  So, implicitly, $P_j$ now denotes 
$D_j \cup D^j_{i,j}$ $=$ 
$\{3,4,5,6,7,8,9,a+1,a+2, \ldots,a+b\}$.  $p_j$ computes $O_{n}(P_j)$ $\cup$ $N(O_{n}(P_j),P_j)$ $=$
$\{3,4\}$ (assuming appropriate tie-breaking).  Moreover, it can be seen that
this satisfies (\ref{neccessary}).  So, $p_j$ sends 
$\{3,4\}$ $\setminus$ $(D^j_{i,j} \cup D^j_{j,i})$ $=$ 
$\{3,4\}$ $\setminus$ $\{3,6\}$ $=$ $\{4\}$ and updates  
$D^j_{j,i}$ to $\{4\}$.

\item $p_i$ will receive these points and updates 
$D^i_{j,i}$ to $\{4\}$ (currently 
$D^i_{i,j} = \{3,6\}$).  So, implicitly, $P_i$ now denotes 
$D_i \cup D^i_{j,i}$ $=$ 
$\{0.5,3,4,6,10,11, \ldots, a\}$.  $p_i$ computes $O_{n}(P_i)$ $\cup$ $N(O_{n}(P_i),P_i)$ $=$
$\{0.5,3\}$.  Moreover, it can be seen that
this satisfies (\ref{neccessary}).  So, $p_i$ sends 
$\{0.5,3\}$ $\setminus$ $(D^i_{i,j} \cup D^i_{j,i})$ 
$=$ $\{0.5,3\}$ $\setminus$ $\{3,4,6\}$ $=$ $\{0.5\}$ and updates 
$D^i_{i,j}$ to $\{0.5,3,6\}$.

\item $p_j$ will receive these points and updates 
$D^j_{i,j}$ to $\{0.5,3,6\}$ (currently 
$D^j_{j,i} = \{4\}$).  So, implicitly, $P_j$ now denotes 
$D_j \cup D^j_{j,i}$ $=$ 
$\{0.5,3,4,5,6,7,8,9,a+1,a+2, \ldots,a+b\}$.  $p_j$ computes $O_{n}(P_j)$ $\cup$ $N(O_{n}(P_j),P_j)$ $=$
$\{0.5,3\}$.  Moreover, it can be seen that
this satisfies (\ref{neccessary}).  So, $p_j$ sends 
$\{0.5,3\}$ $\setminus$ $(D^j_{i,j} \cup D^j_{j,i})$ 
$=$ $\{0.5,3\}$ $\setminus$ $\{0.5,3,4,6\}$ $=$ $\emptyset$. {\em i.e.} 
nothing is sent.

At this point the algorithm has terminated, both sensors are waiting
for an event to occur and there are no messages in flight.  $P_i$ and
$P_j$ denote $\{0.5,3,4,6,10,11, \ldots, a\}$
and $\{0.5,3,4,5,6,7,8,9,a+1,a+2, \ldots,a+b\}$, respectively.  Therefore
$O_{n}(P_i)$ $=$ $\{0.5\}$ $=$ $O_{n}(P_j)$, which in turn, equals the
correct answer $O_{n}(D)$.  Observe that the total amount of communication
(data points sent) was 4.  The naive approach which
centralized all the data on either $p_i$ or $p_j$ requires
$\min\{a-6,b+5\}$ communication.  For large $\min\{a,b\}$, the distributed 
algorithm requires 
much less communication.
\end{enumerate}

\subsection{The Algorithm}

To translate the previous example into a formal algorithm for general $R$ (satisfying
the anti-monotonicity and smoothness axioms), we must provide some definitions generalizing
the role of $N(.,.)$. 
Given $x \in \mathbb{D}$ and finite $P \subseteq \mathbb{D}$, $Q_1 \subseteq P$ is called a {\em support set of} 
$x \in \mathbb{D}$ {\em over} $P$ if $R(x,P) = R(x,Q_1)$.  Intuitively, all other points from $P$ can be discarded without
affecting the rank of $x$.  Using cardnality and the tie-breaking mechanism discussed earlier ($\prec$ a total linear order on $\mathbb{D}$), 
we can define a unique, smallest support set of $x$ with respect to $P$, denoted  $[P|x]$.\footnote{Formally, given
$Q_1, Q_2 \subseteq P$ support sets of $x$ with respect to $P$, $Q_1$ is smaller than $Q_2$ if $|Q_1| < |Q_2|$ or ($|Q_1| = |Q_2|$ and $Q_1$
is lexicographically smaller than $Q_2$ with respect to $\prec$).}
Given $Q \subseteq P$, let $[P|Q]$ denote $\bigcup_{x\in Q}[P|x]$.  In the previous example, $R$ was the distance to the nearest 
neighbor, so, $[P|x]$ equaled $N(x,P)$ using $\prec$ to break ties.  

With these more general definitions, we define a set $Z_j \subseteq P_i$ to be sufficient for $p_j$ if 

\begin{equation}
\label{eq:fixed-point}
\left(O_{n}(P_i) \cup [P_i|O_{n}(P_i)]\right) \cup \left([P_i|O_{n}(D^{i}_{i,j}\cup D^{i}_{j,i} \cup Z_j)]\right) \subseteq Z_j.
\end{equation}

\noindent Due to the broadcast nature of wireless sensor network communication, $p_i$ cannot send points to an single
immediate neighbor without the other neighbors receiving them as well.  In light of this, the algorithm accumulates all points 
(tagged with receipiant IDs)
to be sent to all immediate neighbors in a single packet, $M$.  When an immediate neighbor, $p_j$, receives $M$, the neighbor
extracts those points that are tagged with ID $j$.  If no points are tagged as such, $p_j$ does not regard receipt of $M$ as an
event. 

$p_i$ detects an event if one of the following occurs: (i) the algorithm is
initialized, (ii) $D_i$ changes, (iii) $M$ is received and contains points tagged with $i$ ({\em i.e.} points are received from a neighbor), 
or (iv) a link goes up/down causing $p_i's$ immediate neighborhood to change however, algorithm correctness requires that we assume
the network remain connected).  In response, $p_i$ carries out the following
algorithm whose pseudo-code is given in {\em Global Outliers Detection Algorithm} figure.  First, $P_i$ is updated accounting for 
all $p_j$ from which points were recived in $M$.  Only points not already in $P_i$ are added to $D^{i}_{j,i}$.  
The first two steps
in the main for-loop (``For each $j \in \Gamma_i$, do'') compute a $Z_j$ satisfying (\ref{eq:fixed-point}), although
the result is not guaranteed to be the smallest set to do so.  The ``If....then'' in the main for-loop
tests whether there are any points found sufficient for $p_j$ that $p_i$ cannot already be sure $p_j$ has, {\em i.e.} points in
$Z_j$ but not in $D^i_{j,i} \cup D^i_{i,j}$.  If any such points are found they are added to $M$ along with their recipiant ID $j$.

\begin{figure}

\begin{algorithm}

\caption{Global Outlier Detection}

\begin{algorithmic}

\STATE

\STATE -- Set $M = \emptyset$, and, update $P_i$ accounting for all neighbors $p_j$ from which points were recieved.  For each point $x$
recieved from $p_j$, do the following.  If $x$ is not already $P_i$, then add $x$ to $D^i_{i,j}$.
 
\STATE -- For each $j \in \Gamma_i$, do 

\STATE -- -- Set $Z_j$ $=$ $O_{n}(P_i) \cup [P_i|O_{n}(P_i)]$.

\STATE -- -- Repeat until no change: $Z_j = Z_j \cup [P_i|O_{n}(D^{i}_{i,j}\cup D^{i}_{j,i} \cup Z_j)].$

\STATE -- -- If $Z_j \setminus (D^i_{i,j} \cup D^i_{j,i})$ is non-empty, then

\STATE -- -- -- Append $\langle j, Z_j \setminus (D^i_{i,j} \cup D^i_{j,i}) \rangle$ to $M$. 

\STATE -- -- -- Add points in $Z_j \setminus (D^i_{i,j} \cup D^i_{j,i})$ to $D^i_{i, j}$.

\STATE -- -- End If.

\STATE -- End For.

\STATE -- If $M$ is non-empty, broadcast it to all sensors in $\Gamma_i$.

\end{algorithmic}

\end{algorithm}

\end{figure}

\subsection{Streaming Data and Peer Addition/Deletion}

In our experiments we assume a sliding window model (based on time) in processing the data stream arriving at each sensor.
To do so, we assume each point is time-stamped when sampled by the sensor.  Under the assumption that the sensor clocks 
are synchronized sufficiently well, sensor $p_i$ deletes all points in $P_i$ (regardless of where they were originally sampled) 
once thier time-stamp indicates they are no longer in the window.  

The algorithm can be easily modified to accomidate the addition of sensors during operation.  All that is required to do so is
treat the arrival of a new sensor as an event for the new sensor and for all its immediate neighbors.  
The algorithm can also be modified to accomidate the removal of sensors ({\em e.g.} when their battery is depleted) assuming
that the network remains connected.  In the sliding window model, a simple strategy is to merely
allow points that originated with the removed sensor to age out of the window at the expense of tolerating, strictly speaking, an inaccurate 
result until this happens.  A more general and complex solution is to propogate messages into the network causing sensors to
explicitly delete those points that originated with the removed sensor.  We leave 
the details of this approach to future work.  

\subsection{Algorithm Correctness}

The correctness of the algorithm can be proved in the following sense: 
if the data and network links remain static (and the network is connected), then communication will eventually 
stop at which point all sensors' outlier estimate will equal $O_n(D)$.
It is important to emphasize that this does {\em not} mean the algorithm cannot
handle dynamic data or network links.  Merely that, upon such a change, the algorithm
will respond and converge on the correct answer.  But, naturally, such
convergence is gauranteed only if the data and network remain static long enough.

It is easy to see that, barring data or network change, the algorithm will always
terminate.  So, the proof proceeds in two steps. First, upon
termination, all sensors have the same outlier estiamtes and support 
(Theorem \ref{thm:termination}).  Second, the consistent outlier estimates shared by all sensors 
is indeed the correct one (Theorem \ref{thm:correctness}).  Proofs of Theorems \ref{thm:termination} and 
\ref{thm:correctness} are provided in Appendix \ref{appendix:proofs}.

\begin{thm}
\label{thm:termination}
Assuming a connected network, if for all sensors $p_i$: $D_i$ and $\Gamma_i$ do not change, then upon termination of 
the algorithm all sensors' outlier estimates and supports agree: for all $p_i, p_j$: 
$O_n(P_i) = O_n(P_j)$ and $[P_i|O_n(P_i)] = [P_j|O_n(P_j)]$.
\end{thm}

\begin{thm}
\label{thm:correctness}
Assuming a connected network, if for all sensors $p_i$: $D_i$ and $\Gamma_i$ do not change, then upon termination of the 
algorithm, all sensors' outlier estimate will be correct: for all $p_i$: 
$O_n(P_i) = O_n(D)$.
\end{thm}

\noindent {\bf Comments:} {\bf 1)} Theorem \ref{thm:termination} holds without the smoothness axiom, hence, for any
anti-monotonic $R$, upon convergence, all sensors will
agree on their outlier estimate and support.  However, without the smoothness axiom,
Theorem \ref{thm:correctness} does not hold, {\em i.e.} the consistent outlier estimates
might not be the correct one.  There are counter-examples which show how an 
anti-monotonic, but not smooth $R$ cause the algorithm to terminate without all
sensors agreeing upon the correct set of outliers.   

{\bf 2)} For an arbitrary $R$, it is not clear how to efficiently compute $[P|x]$ and
we do not address the issue.
However, efficient computation is straight-forward for the $R$ we
consider in our experiments: average distance to the $k^{th}$ nearest neighbor.  

\section{Semi-Global Distributed Outlier Detection Algorithm}
\label{sec:localized}

It can be desirable 
for sensors to find outliers only with respect to the data contained in nearby sensors, rather than
the entire network.   
In this section, we describe how to modify the global
outlier detection algorithm to act in a semi-global manner.  Under this approach, 
each sensor computes 
outliers only from within those points sampled in its spatial proximity.  

To account for spatial locality, we use {\em hop distance}: the number of hops
bewteen two sensors along their shortest path in the underlying communication network.  
Given integer $d$ and sensor $p_i$, let $D_i^{\leq d}$
denote the union of all $D_j$ such that $p_j$ and $p_i$ have hop distance no greater
than $d$.  The semi-global outlier detection problem requires each $p_i$ to compute 
$O_n(D_i^{\leq d}$).  Setting $d$ to infinity yields the 
global outlier detection problem discussed earlier.  

To account for hop distance, each data point $x$ has an additional field $x.hop$
(at birth $x.hop$ is set to zero).  Let $x.rest$
denote all the remaining fields -- these are the ones used by the rating
function $R$.
Given a set of points $Q$, for $0 \leq h \leq d$, let $Q^{\leq h}$ be the 
set of points $x \in Q$ with $x.hop \leq h$.  Let $[Q]^{min}$ be the result of replacing 
all points that differ only in their hop field by the point with the smallest
hop field.  For example, consider $Q = \{w,v,x,y,z\}$ where
$w.rest = v.rest$, $x.rest = y.rest = z.rest$, and $v.rest \neq x.rest$.  If
$w.hop < v.hop$ and $x.hop < y.hop, z.hop$, then 
$[Q]^{min}$ $=$ $\{w,x\}$.

\subsection{Semi-Global Outlier Detection Algorithm}

The basic idea is that each 
sensor $p_i$ will run the global outlier detection algorithm
over only those points arising on sensors within $d$ hops.  At first glance, 
the following simple modification of the global algorithm seems adequate.  Before $p_i$ sends a 
copy of a point, $x$, to its neighbors, it first increments 
$x.hop$ and sends only if $x.hop \leq d$.  Unfortunately, such a simple modification
will not work.  It does not take into account the fact that $x$ should not have
any effect on the outlier determination process of sensors $p_j$ whose distance
from $p_i$ is more than $d-x.hop$.  Examples can be demonstrated wherein this omission causes
an incorrect overall result. 

To avoid this problem, $p_i$ must partition $P_i$ into $d$ parts:
$P_i^{\leq h}$ for $0 \leq h \leq d-1$.  For each, in essence, the global outlier detection
algorithm is applied.  Upon detecting an event (defined as before),
$p_i$ carries out the following algorithm
whose pseudo-code is given in the {\em Semi-Global Outlier Detection Algorithm} figure. 

First, $P_i$ is updated accounting 
for all $p_j$ from which points were recived in $M$.  Because of the hop fields, the update step is somewhat more 
complicated than that of the Global Outlier Detection Algorithm.  A point $x$ from $p_j$ is added to $D^{i}_{j,i}$ 
if there does not exist $y \in P_i$ with $x.rest = y.rest$ ($x$ does not already
appear in $P_i$).  Or, if there does exist $y \in P_i$ with $x.rest = y.rest$, but $x.hop < y.hop$, then
$x$ replaces $y$ in $P_i$ (updating as needed $D_i$ and $D^i_{f,i}$ for 
each $f \in \Gamma_i$).  Note, there cannot be more than one $y$ with the same rest fields as $x$ since all but the
point with the smallest hop would have been removed earlier.  

Next, for each neighbor $p_j$ and each $0 \leq h \leq d-1$, a set $Z_j^h$ is computed which satisfies (\ref{eq:fixed-point}) with
$Z_j$, $P_i$, $D^i_{i,j}$, and $D^i_{j,i}$ replaced by $Z_j^h$, $P_i^{\leq h}$, $D^{i,\leq h}_{i,j}$, and
$D^{i,\leq h}_{j,i}$, respectively.  This computation is done by the first steps inside the nested for loops.
Then, the hop field for each point in $Z_j^h$ is incremented in preparation for sending to $p_j$.  

Once all $0 \leq h \leq d-1$ have been processed for $p_j$ (the inner for loop completes), $Z_j^1 \cdots Z_j^{d-1}$ are
unioned and redundancies are eliminated.  For any pair of points $x,y$ in $\bigcup_{h=0}^{d-1}Z_j^h$, if $x.rest = y.rest$ and
$x.hop < y.hop$, then $y$ is dropped (this action is signified by the `min' superscript in the step immediately after the inner 
for loop).  Then, all points $x$ are removed from $Z_j$ if there exists a point $y$ in $(D^{i}_{i,j} \cup D^{i}_{j,i})$ with 
the same rest fields but $y.hop \leq x.hop$.  If the resulting $Z_j$ is non-empty, then these points are be added to 
$M$ (along with ID $j$) for broadcast to neighbors.  And,
$D^{i}_{i,j}$ is updated by adding the points in $Z_j$.

\begin{figure}

\begin{algorithm}

\caption{Semi-Global Outlier Detection}

\begin{algorithmic}

\STATE

\STATE -- Set $M = \emptyset$, and, update $P_i$ accounting for all neighbors $p_j$ from which points were recieved.  For each point $x$
recieved from $p_j$, do the following.  If there does not exist $y$ in $P_i$ with $x.rest$ $=$ $y.rest$, then add $x$
to $P_i$ (and update $D^i_{i,j}$), otherwise if $x.hop < y.hop$, then replace $y$ with $x$ in $P_i$ (updating as needed $D_i$ and $D^i_{f,i}$ for 
each $f \in \Gamma_i$).
 
\STATE -- For each $j \in \Gamma_i$, do 

\STATE -- -- For $h = 0$ to $d-1$ 

\STATE -- -- -- Set $Z_j^{h}$ $=$ $O_{n}(P_i^{\leq h}) \cup [P_i^{\leq h}|O_{n}(P_i^{\leq h})]$.

\STATE -- -- -- Repeat until no change: $Z_j^h = Z_j^h \cup [P_i^{\leq h}|O_{n}(D^{i,\leq h}_{i,j}\cup D^{i,\leq h}_{j,i} \cup Z_j^h)].$

\STATE -- -- -- Increment the hop field for each point in $Z_j^h$.

\STATE -- -- End For.

\STATE -- -- Set $Z_j = \left[\bigcup_{h=0}^{d-1}Z_j^h\right]^{min}$.

\STATE -- -- Remove points $x$ from $Z_j$ such that there exists $y \in (D^{i}_{i,j} \cup D^{i}_{j,i})$ with $x.rest = y.rest$ and  $y.hop \leq x.hop$.

\STATE -- -- If $Z_j$ is non-empty, then

\STATE -- -- -- Append $\langle j, Z_j \rangle$ to $M$. 

\STATE -- -- -- Update $D^{i}_{i,j}$ by adding the points in $Z_j$.

\STATE -- -- End If.

\STATE -- End For.

\STATE -- If $M$ is non-empty, broadcast it to all sensors in $\Gamma_i$.

\end{algorithmic}

\end{algorithm}

\end{figure}

\section{Performance evaluation}

\subsection{Experimentation setup}

We used the SENSE wireless sensor network simulator \cite{sense:2004} to evaluate the
performance of the global and semi-global outlier detection algorithms.
Specifically, we analyzed the following metrics: (1) the accuracy of the algorithms
in detecting outliers; (2) the average amounts of total energy, transmission energy,
and receive energy consumed per node per sampling period; and (3) the minimum and
maximum amount of energy consumed in the network. We observed both the global and
semi-global outlier detection algorithms to be highly accurate as nodes converged
upon the correct results approximately 99\% of the time. We attribute any detection
error to dropped packets. Since average detection accuracy was consistent across
all simulation parameters, we did not include any accuracy-related plots in this
manuscript.

Various scenarios were used to analyze the performance of our algorithms. First, we
compared our algorithms' energy usage against that of a purely centralized outlier detection
algorithm. Here, all nodes periodically sent their sliding window contents to a central
node which detected outliers based on the unioned data sets and returned the
outliers back to the nodes. For simplicity, we configured the centralized algorithm
to calculate only global outliers since for this algorithm, energy usage is
independent of whether global or semi-global outliers are detected. 
i
Also, all algorithms (including
the centralized solution) were evaluated using the following two outlier ranking
functions ($R$): \textit{distance to nearest neighbor} (NN) and \textit{average
distance to k nearest neighbors} (KNN).

We chose to use a centralized algorithm for our comparison because, to the best
of our knowledge, there exist no comparable distributed solutions for WSN outlier detection.
We find such a comparison to still be valid as many WSN deployments continue to employ
centralized configurations, citing ease of administration as well as maintenance of a
single (and "`standard"') point of interface with the growing number of applications
and systems in the sense-and-respond computing domain. Such reasons, however, do not
preclude the utility of a distributed algorithm such as ours, since it remains very
useful as a general data processing solution for a wide range of applications, whether they
are centralized are not.

For our data sets, we used real-world recorded sensor data streams from \cite{sensordata},
in which distributed data samples are both spatially and temporally correlated. The
data set we used was composed of series of data samples describing environmental
phenomena such as heat, light, and temperature from 53 sensors which periodically
transmitted individual data samples to a central base station. The data set did
contain missing data points, which to the best of our knowledge was largely due to
packet loss. Hence, we replaced missing data points with the average values of the
data points within sliding windows preceding the missing points. This helped retain
the temporal trends of the data streams. The data points we used contained the following
features: (1) ID of the sensor that produced the data point; (2) epoch (sequential number
denoting the data point's position in the sensor's entire stream); (3) data value (we
specifically used temperature); and (4) x,y location coordinates. We used the data
points' temperature value and location coordinates as inputs into the outlier rating
functions. The location coordinates can represent either the place of
measurements or an estimate of a position of a target or some other spatial
information. It is important to note that these coordinates are a part of
the data on which our algorithm works in the example. They might suffer errors,
and become anomalous, just as would any other attribute of the data, due to
an inaccurate initialization, power degradation, or a transmission
error. The algorithm itself, however, would work the same regardless if such
coordinates are given or not.

We originally simulated two networks based on the coordinates of the sensors in the
data set: the first of size 32 nodes (which included a uniformly random sampling of
the full network) and the second of size 53. The purpose of simulating two networks
was to examine how well the algorithm scaled with the size of the network. We found
that the as the network size increased, the performance benefit of the distributed
algorithms increased in comparison to the centralized algorithms and that performance
trends for different test variables were generally the same. Hence, we did not include
any detailed results associated with the smaller network in this manuscript.

We simulated a terrain of size 50m$\times$50m. Most hardware specifications claim that
a sensor node's transmission range typically reaches up to approximately 250m, when properly
elevated. However, when placed on the ground the reported ranges are much smaller \cite{zuniga04}
and for reliable communication indoor using Crossbow motes, the effective range drops to a
few meters~\cite{szymanski07}. Therefore, we configured all nodes to have a  uniform transmission
range of approximately 6.77m. We also used the hardware energy model based on the Crossbow mote
specifications \cite{motespecs} with a transmit/receive/idle power setting of 0.0159W/0.021W/3e-6W
(assuming a  3V power source). We simulated the wireless transport medium using the 
free-space signal propagation model.

Two protocols were used for routing. For the distributed algorithms, we used simple broadcast (as
opposed to unicast) transmission with promiscuous listening that allowed all nodes
to send data points to all their adjacent neighbors using one transmission. For the
centralized algorithm, we used the well accepted AODV \cite{perkins97ad} wireless
routing protocol for multi-hop communication. We note that a simple end-to-end
acknowledgment mechanism was also used to reinforce reliable communication. While alternative
protocols do exist for more data-centric and energy-efficient communication, our main goal
was to compare the overhead between the algorithms in a straight-forward manner without
having to involve ourselves with balancing various advantages towards either algorithm.

All simulations were run for 1000 seconds of simulated time and were repeated four times
using different random number generator seed values to obtain averaged results. As shown in the following plots,
we collected results for different values of the following algorithm parameters: (1) the
length of the node's sliding window, \textit{w}; and (2) the number of outliers to be reported,
\textit{n}. Additionally, for the distributed localized outlier detection
algorithm, we varied the hop diameter for the localized outlier detection
algorithm for from one to three hops. The
labeling of the data in the plots is as follows: (1) $Centralized$ for results obtained
with the centralized algorithm; (2) $Global-NN$ and $Global-KNN$ for results obtained
using distributed global outlier detection with NN and KNN outlier detection ranking functions ($R$),
respectively; and (3) $Semi-global, epsilon=x$ for all results obtained using distributed
localized outlier detection where $x$ is the value of the hop diameter of the
spatial extent outlier detection. For brevity, we will often refer to the 
different algorithms by these labels.

\subsection{Experimentation results}

\begin{figure}[t]
\begin{center}
\includegraphics[width=7cm]{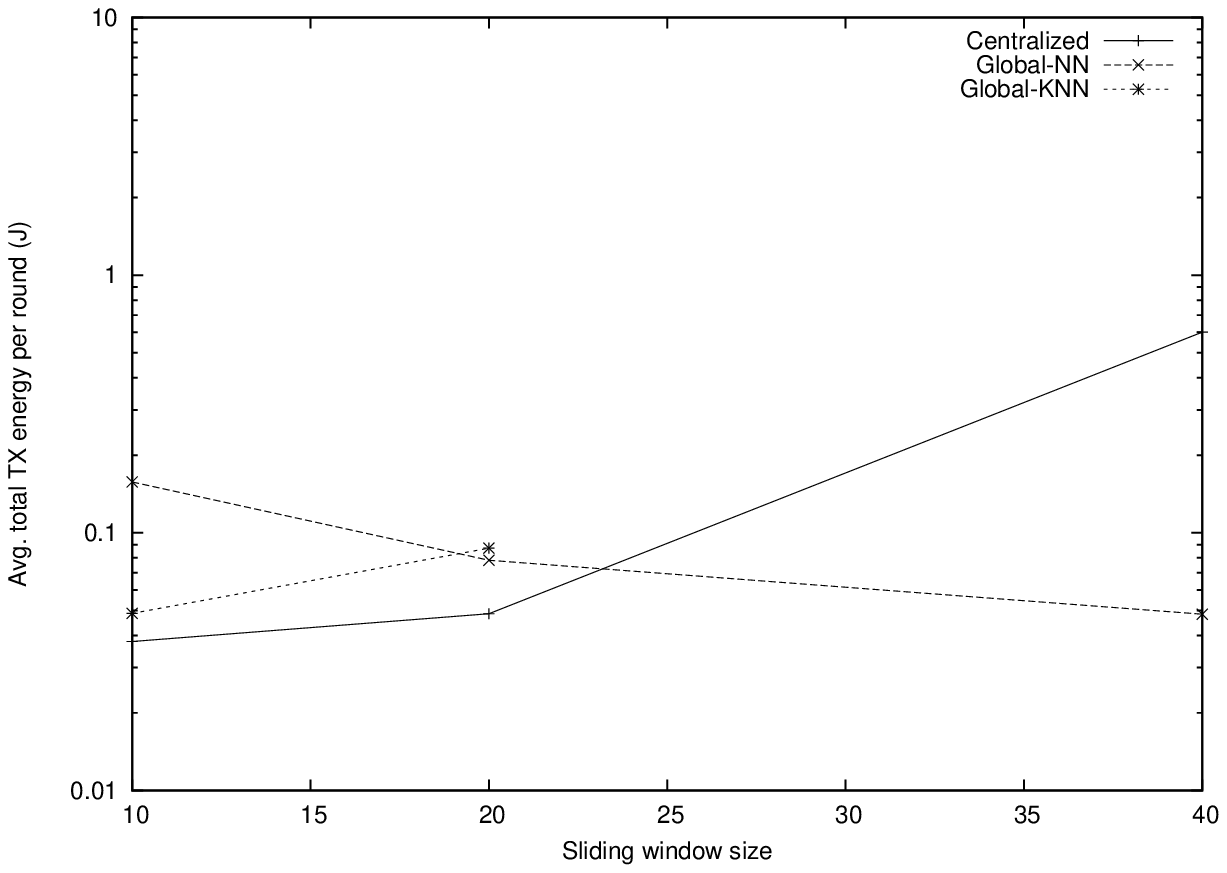}
\includegraphics[width=7cm]{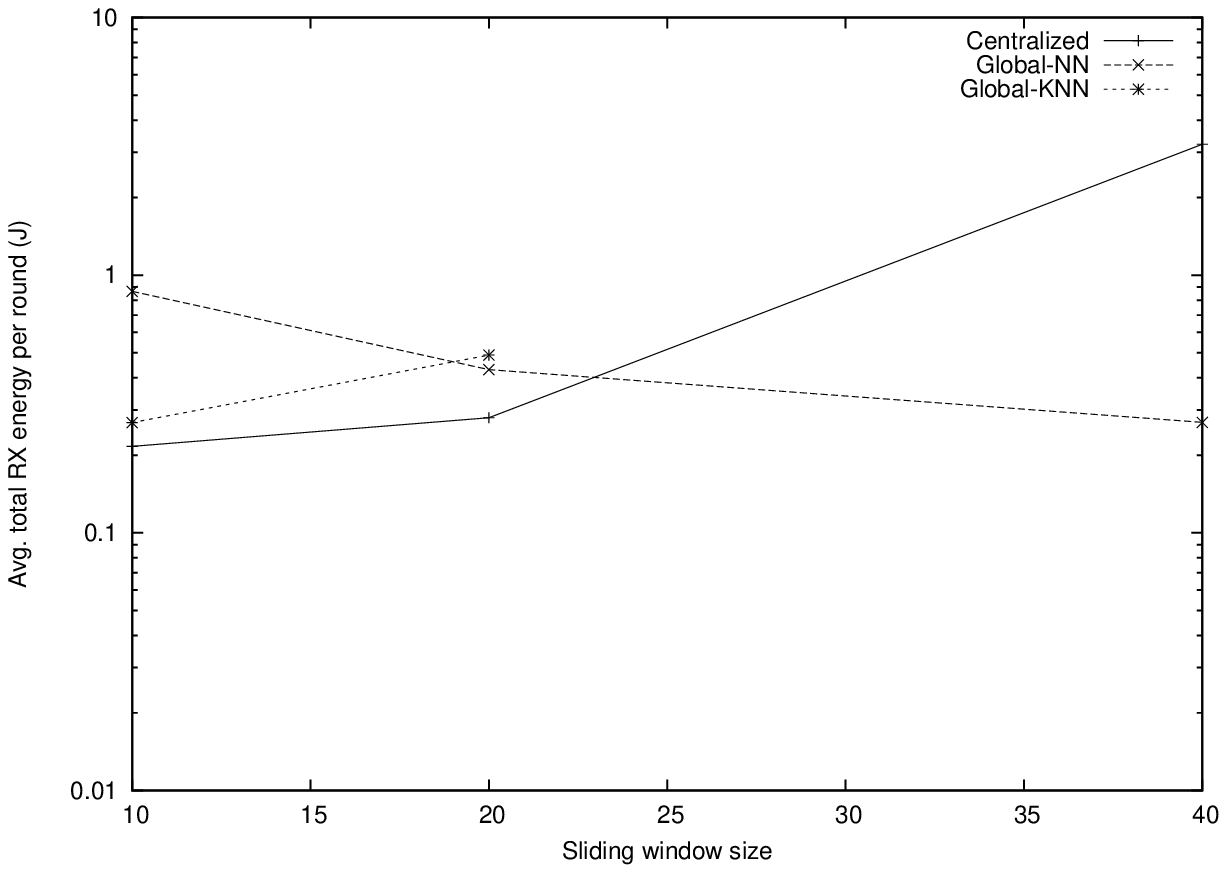}
\caption{Average transmission and receive energy consumed per node per sample interval vs. $w$ ($n$=4, $k$=4)
for global outlier detection.}
\label{fig:global1}
\end{center}
\end{figure}

\subsubsection{Effect of sliding window size}

The plots in Figure \ref{fig:global1} compare the rate
of energy usage of the network between using the centralized algorithm
and the distributed algorithm for global outlier detection as $w$ increases and $n$ and $k$
remain fixed at 4. Here, we show separate plots for transmission (TX) and receive (RX) energy for
the reader who is interested in the disparity between the energy consumption due to different
radio operations. We note that data points are
missing for Global-KNN at $w$=40, due to the inability of our computing resources to
complete simulations for this particular algorithm at the given parameter value. However,
preliminary results based on similar simulations and statistics are shown in \cite{BSWGK:2006}.
Since the trends between both sets of results are nearly identical for the non-missing
data points, it is reasonable to extrapolate the values of the missing points here.

Both figures show that as $w$ increases, Global-NN is the only algorithm that reduces its
energy usage. Figure \ref{fig:global1} shows that Global-NN eventually becomes the most
energy-efficient solution given the domain of $w$. We attribute
Global-NN's reduction in energy usage to an increasing amount of incoming data redundancy
as the size of the sliding window increases. Since Global-NN only uses one supporting
point to determine an outlier, the probability of finding new outliers or supports
with this scheme in each new time interval as the sliding window increases is low.

Regarding the energy consumption of Global-KNN and Centralized, Figure \ref{fig:global1}
reveals trends of increasing energy consumption as $w$ increases.
However, given comparable results in \cite{BSWGK:2006}, we can extrapolate that Global-KNN's
amount of energy consumption is a concave increasing function of
$w$, whereas Centralized's is a convex increasing function, which makes the latter comparatively
less energy-efficient in that it approaches a point of network failure at a higher rate.
In comparison to Centralized, the energy trend for Global-KNN for this data set
indicates that when global outliers are defined by multiple supporting
points (in this case 4), the increasing size of the time interval from which the points
are chosen has a less drastic effect on the number of messages required for the algorithm to converge.
Overall, in cases where a user prefers to use more supporting points and a larger sliding
window to define outliers, energy usage will be higher than using Global-KNN, but it is
still more beneficial to use a distributed solution.

\begin{figure}[t]
\begin{center}
\includegraphics[width=7cm]{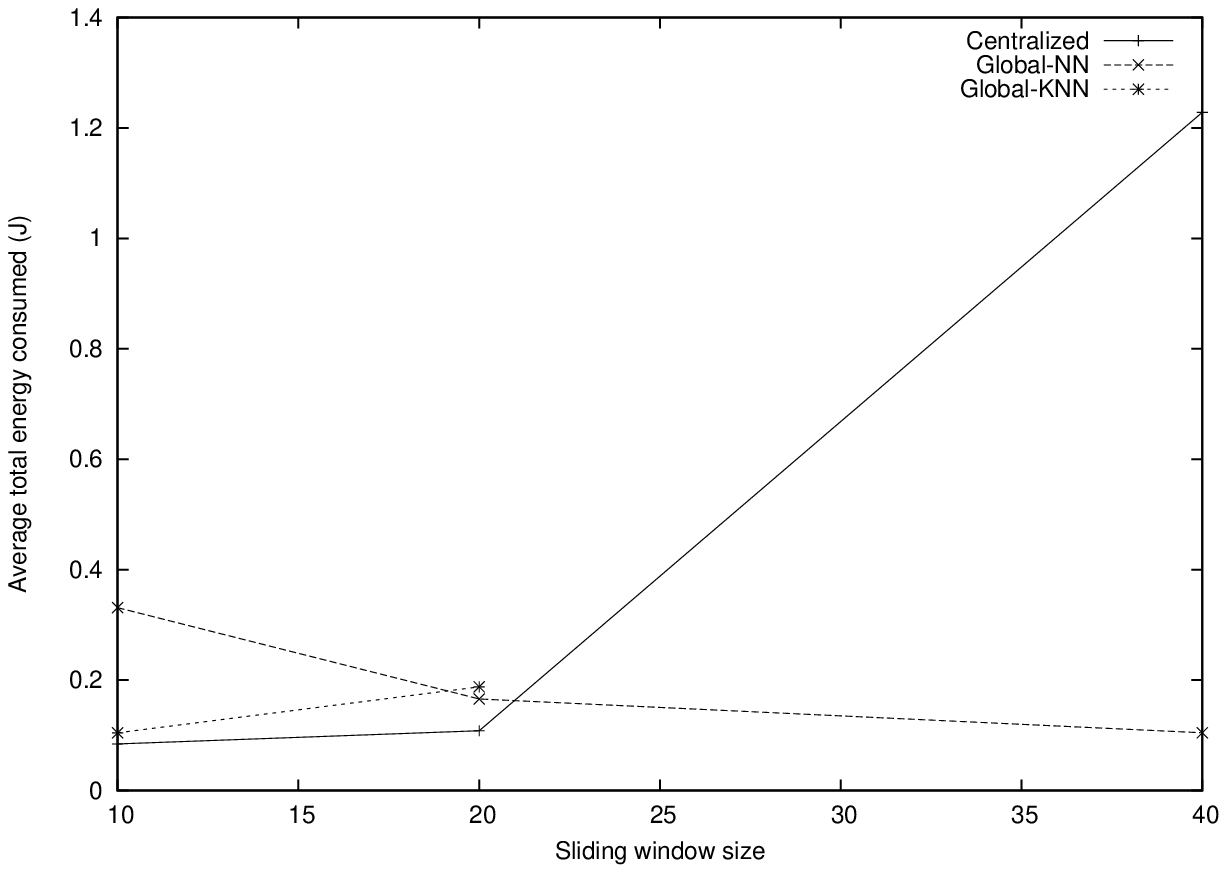}
\includegraphics[width=7cm]{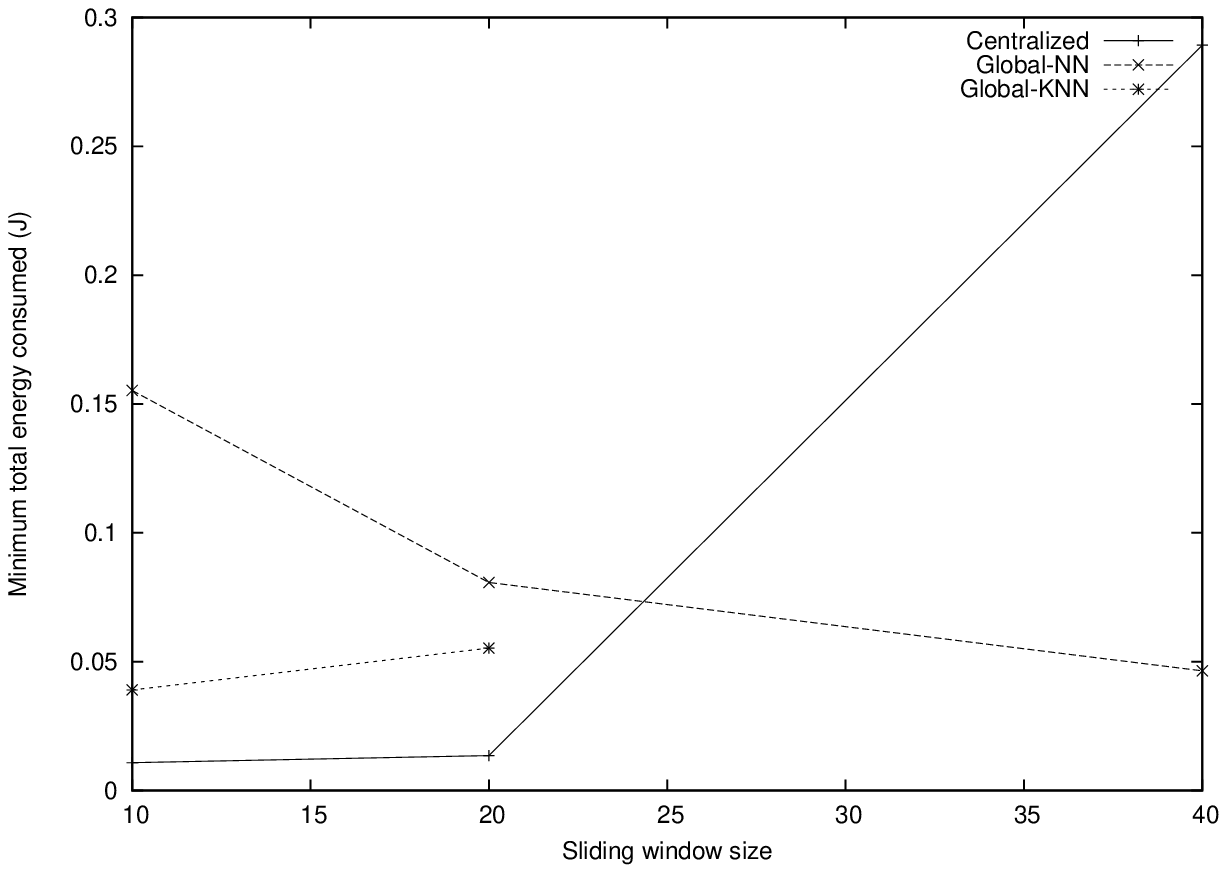}
\includegraphics[width=7cm]{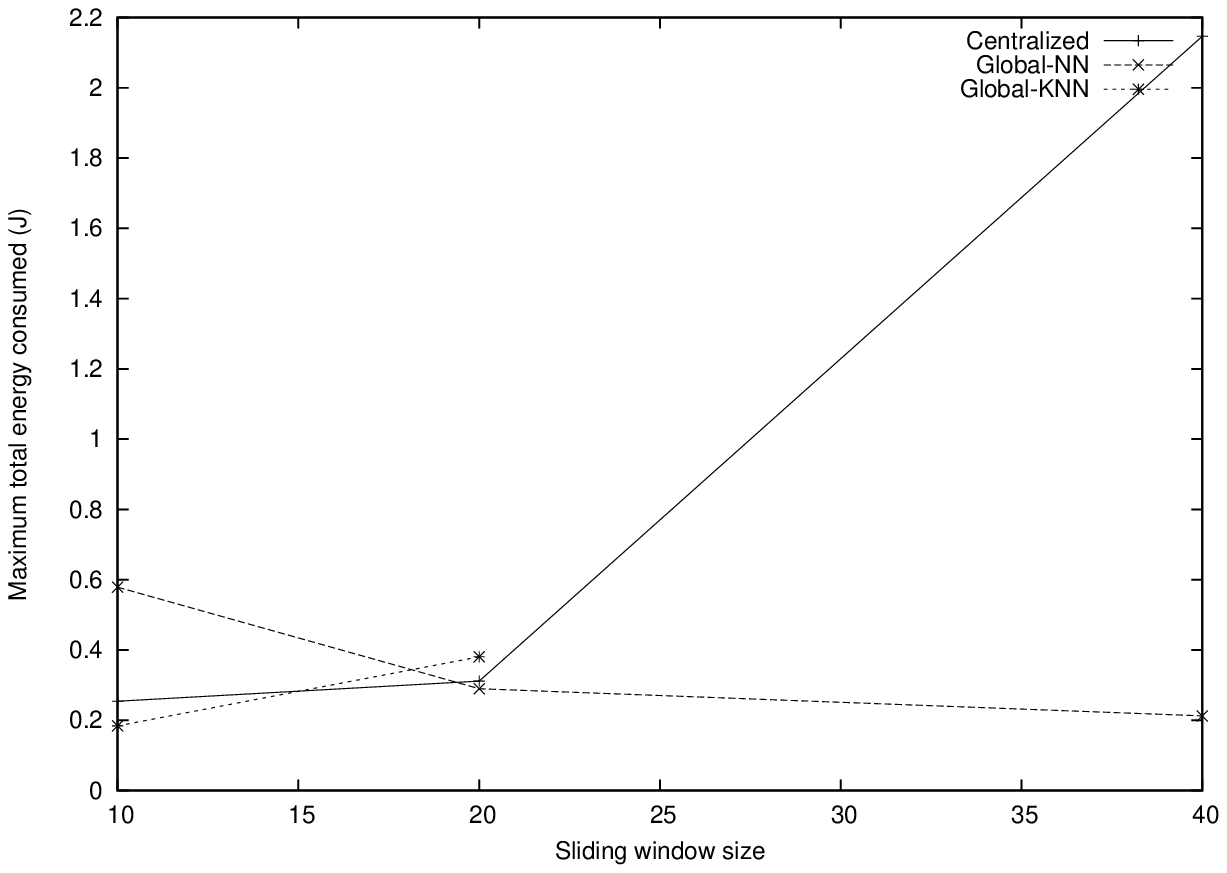}
\caption{Average, minimum, and maximum amount of energy consumed by a node
for global outlier detection.}
\label{fig:min_avg_max}
\end{center}
\end{figure}

Figure \ref{fig:min_avg_max} shows the minimum, average, and
maximum amounts of energy consumption for a sensor node as $w$ increases. Since we limit our focus
primarily to the $range$ of a sensor's energy consumption,
with the intent of analyzing how energy is balanced under the different algorithms, we present
data in terms of total energy consumption. The analysis of TX and RX energy have less value here. Figure \ref{fig:min_avg_max} further accentuates the advantage of using the Global-NN
outlier detection solution over Centralized for large window sizes.
Another observation is that the range of energy consumption for different motes running
the same detection algorithm is larger for the centralized solution than for the
distributed solution. Figure \ref{fig:min_avg_max_norm}
clearly expresses this point by illustrating the values shown previously in Figure
\ref{fig:min_avg_max}, only this time
normalizing the values with respect to the average energy consumption.
For $w$=10, the most energy consuming node consumed nearly
three times more energy than the average node in a centralized algorithm and
less than twice the energy of the average node in both distributed algorithms.

\begin{figure}[t]
\begin{center}
\includegraphics[width=5cm,angle=-90]{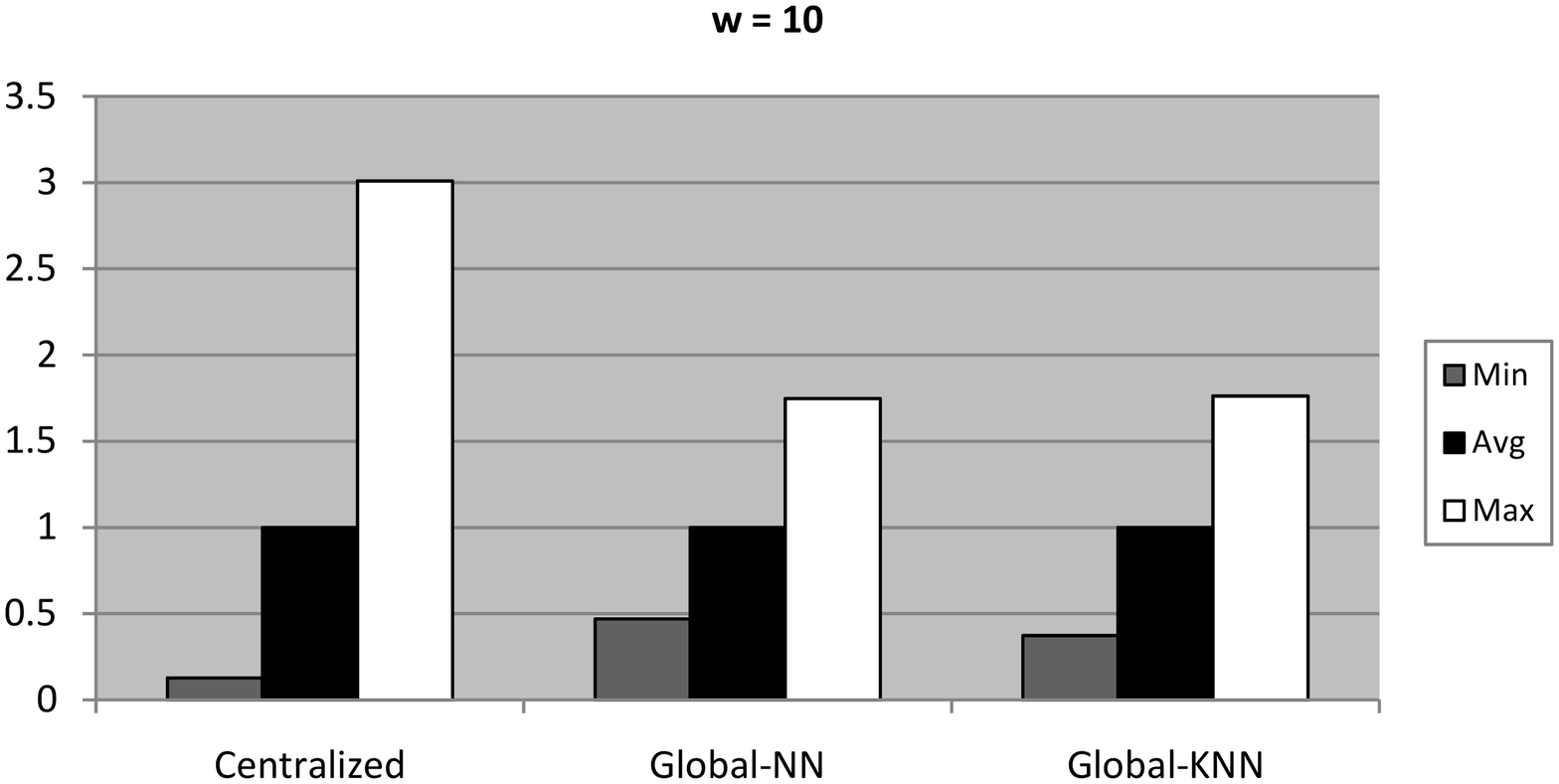}
\includegraphics[width=5cm,angle=-90]{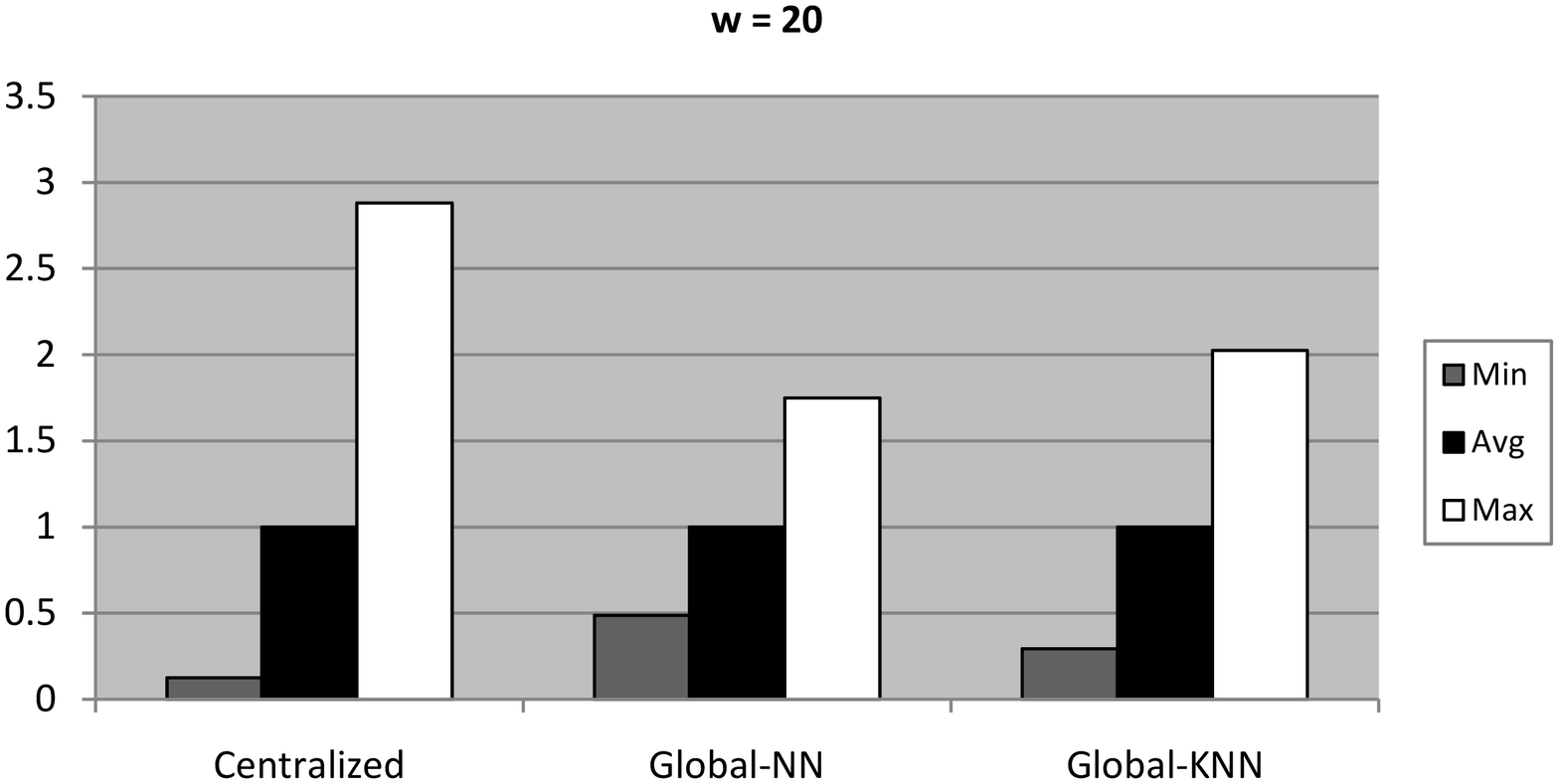}
\includegraphics[width=5cm,angle=-90]{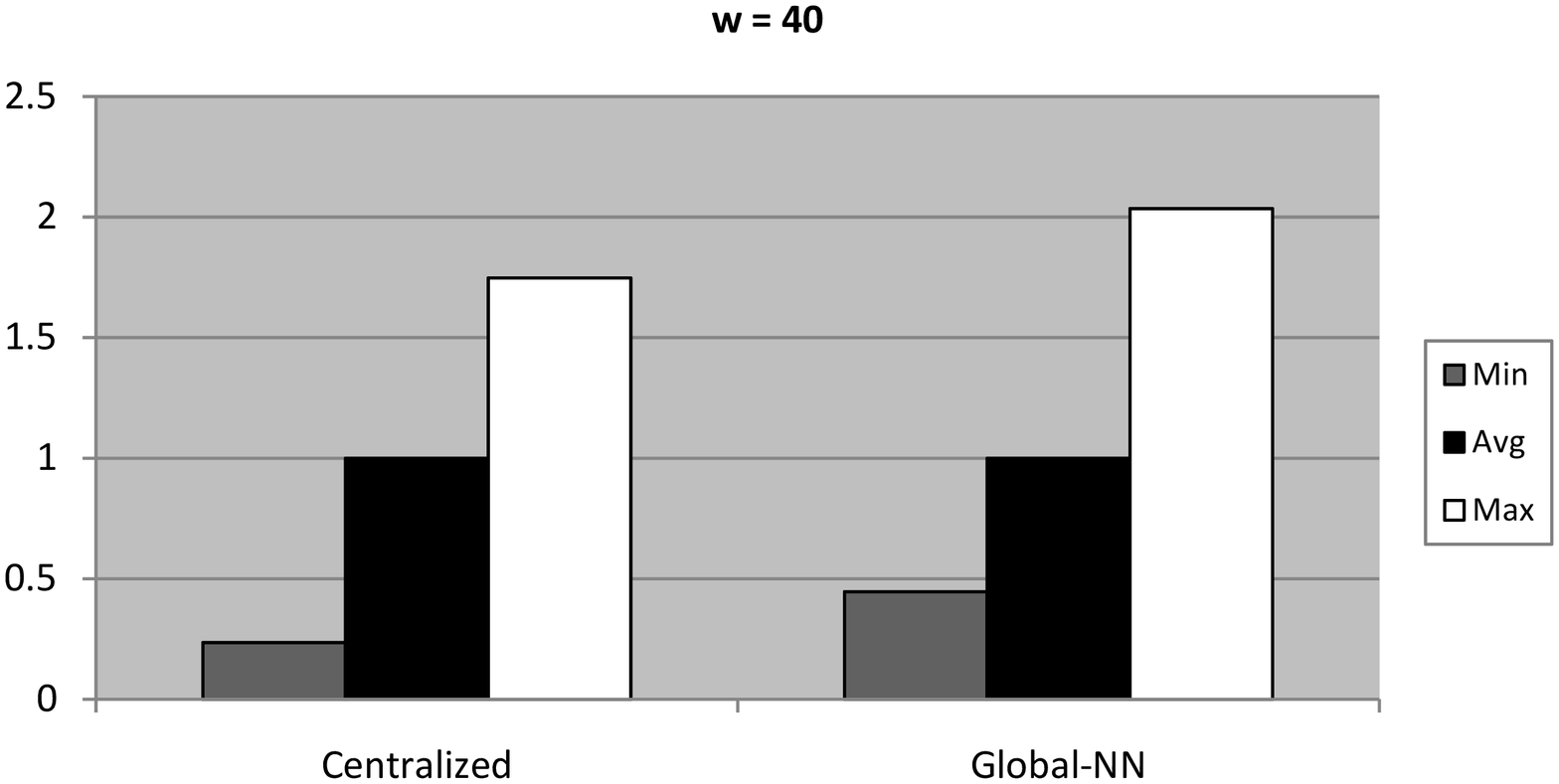}
\caption{Normalized average, minimum, and maximum amount of energy consumed by a node
for global outlier detection.}
\label{fig:min_avg_max_norm}
\end{center}
\end{figure}

For the partial information for $w$=40, the normalized range of energy consumption is actually lower
for the centralized algorithm than for the distributed one. However, referring
back to Figure \ref{fig:min_avg_max}, the average
energy consumption for a node in the centralized case is much higher than that
for the distributed case. Hence, in this case, the normalized maximum value
does not convey the full picture of energy quality of the compared algorithms.

\begin{figure}[t]
\begin{center}
\includegraphics[width=7cm]{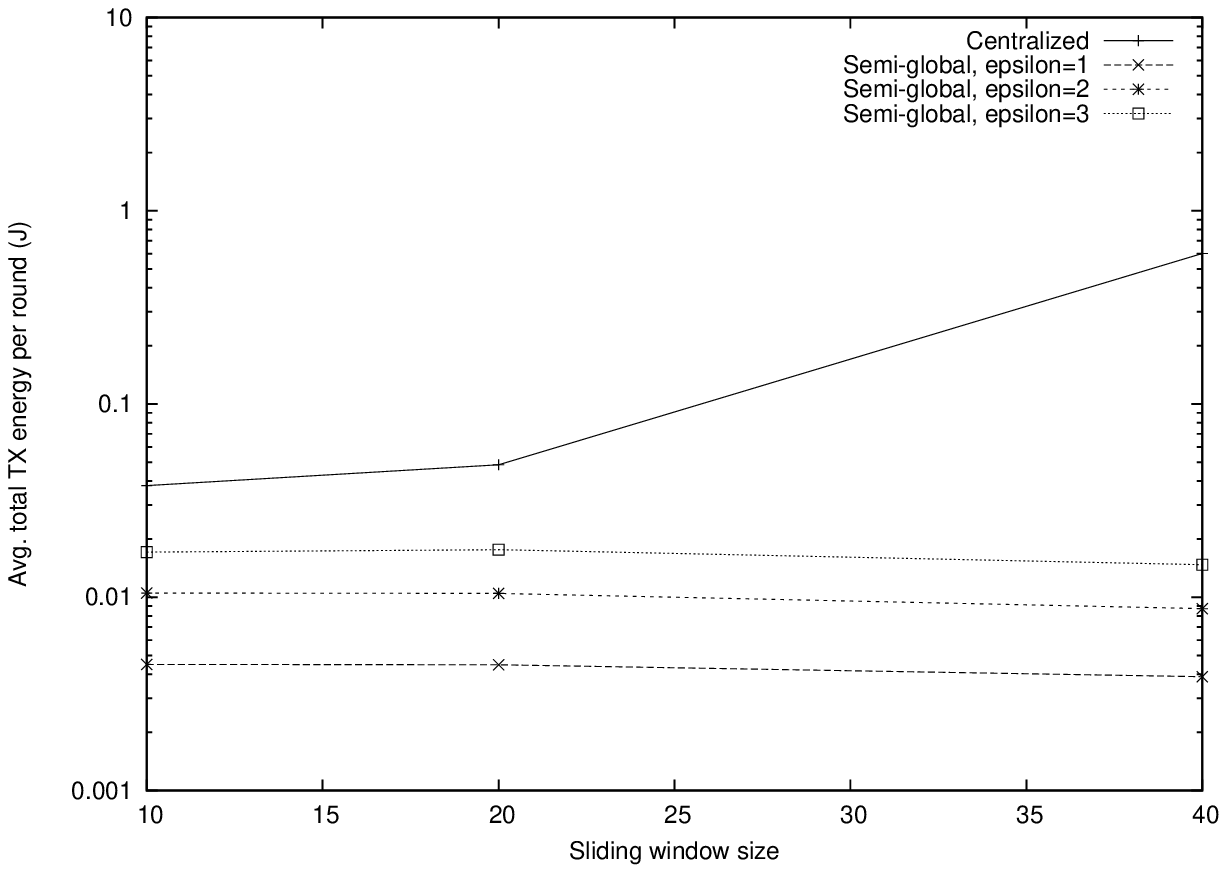}
\includegraphics[width=7cm]{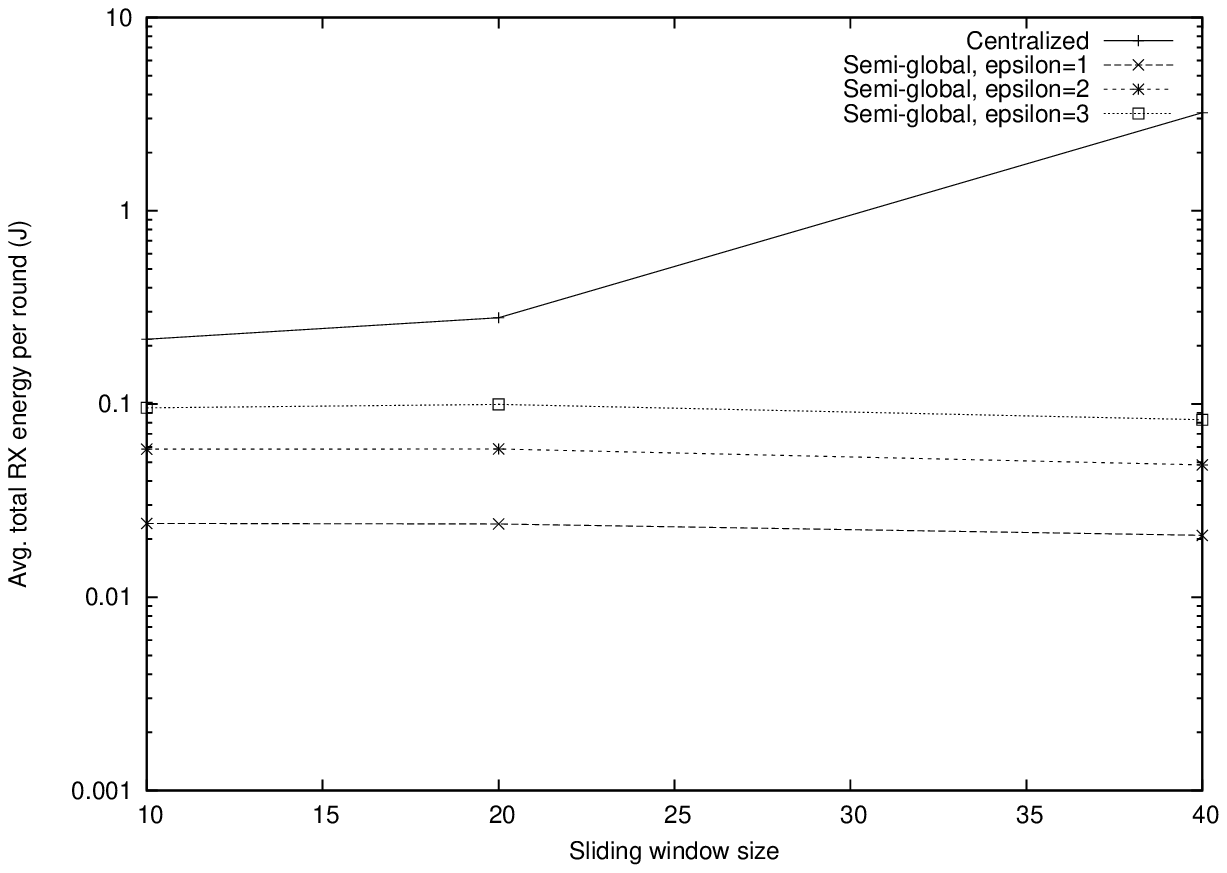}
\caption{Average transmission and receive energy consumed per node per sample interval vs. $w$ ($n$=4)
for localized outlier detection using nearest neighbor outlier detection.}
\label{fig:local1_NN}
\end{center}
\end{figure}

\begin{figure}[t]
\begin{center}
\includegraphics[width=7cm]{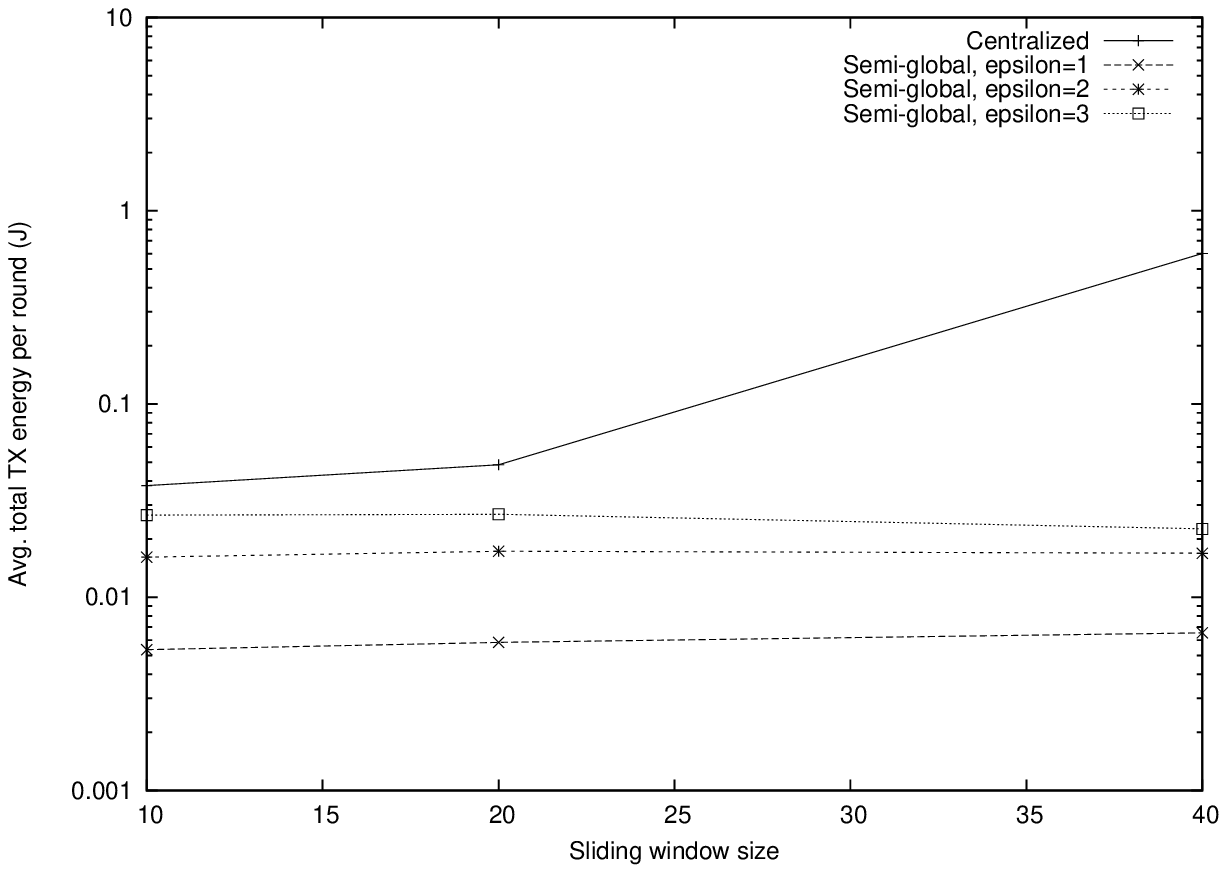}
\includegraphics[width=7cm]{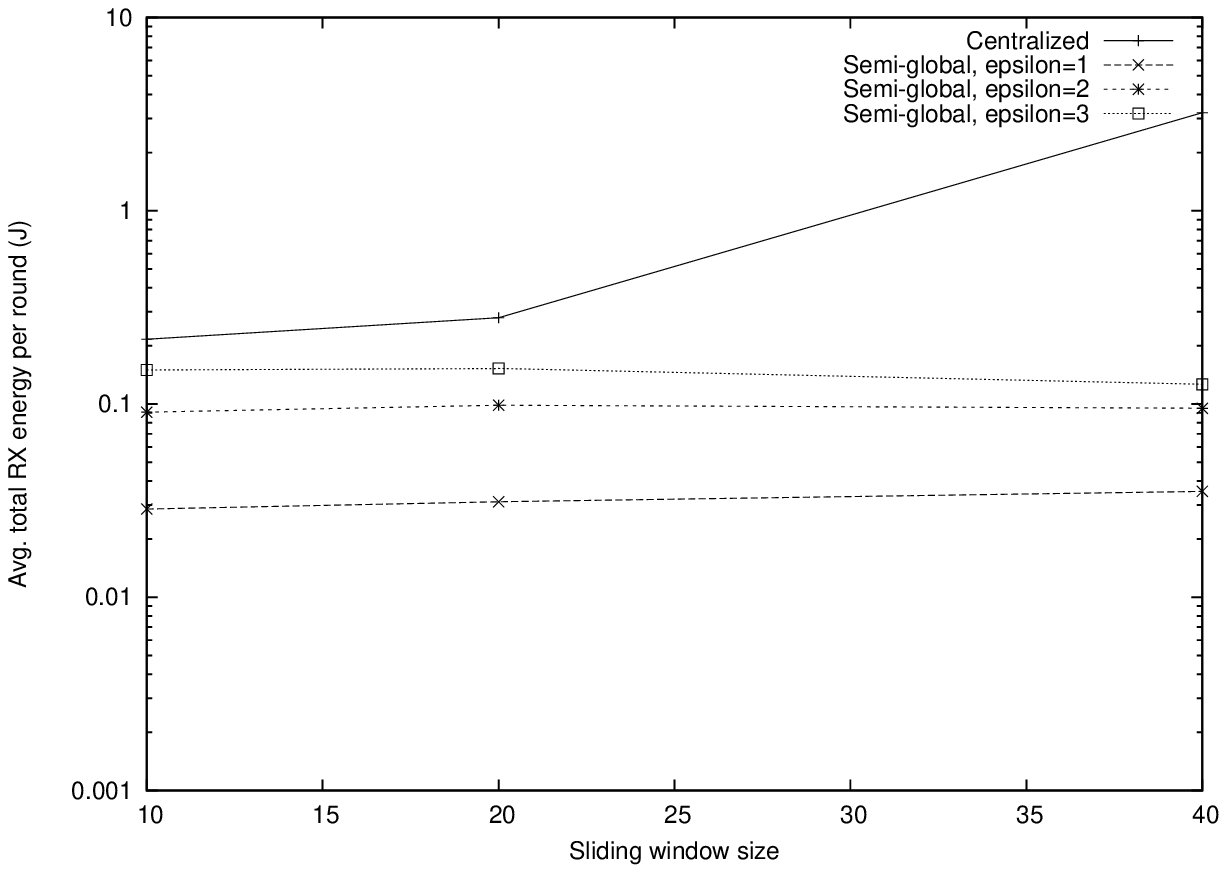}
\caption{Average transmission and receive energy consumed per node per sample interval vs. $w$ ($n$=4, $k$=4)
for localized outlier detection using $k$ nearest neighbor outlier detection.}
\label{fig:local1_KNN}
\end{center}
\end{figure}

The plots in Figure \ref{fig:local1_NN} compare the
rate of energy usage between the centralized algorithm and the distributed algorithm
for localized outlier detection. Since the results of using NN and KNN
outlier detection methods are nearly identical, only results for the former are shown.
Again, the centralized algorithm uses much more energy than the distributed algorithms.
Regarding the distributed localized algorithms, the rate of energy usage increases
along with the values of epsilon. This is expected since as epsilon increases, so
does the message passing overhead as data points travel farther from their place
of origin. The behavior of the distributed algorithm in the localized case for
nearest neighbor outlier detection is similar to that of global case for the same
detection method. Energy usage generally decreases as $w$ increases. As before, we
attribute this behavior to the increasing amount of data redundancy as the size of
the sliding window increases. In general, the extent of the spatial area over which
outliers are defined affects the energy usage trends of the algorithm, but not by
a significant amount.

\subsubsection{Effect of the number of reported outliers}

\begin{figure}[t]
\begin{center}
\includegraphics[width=7cm]{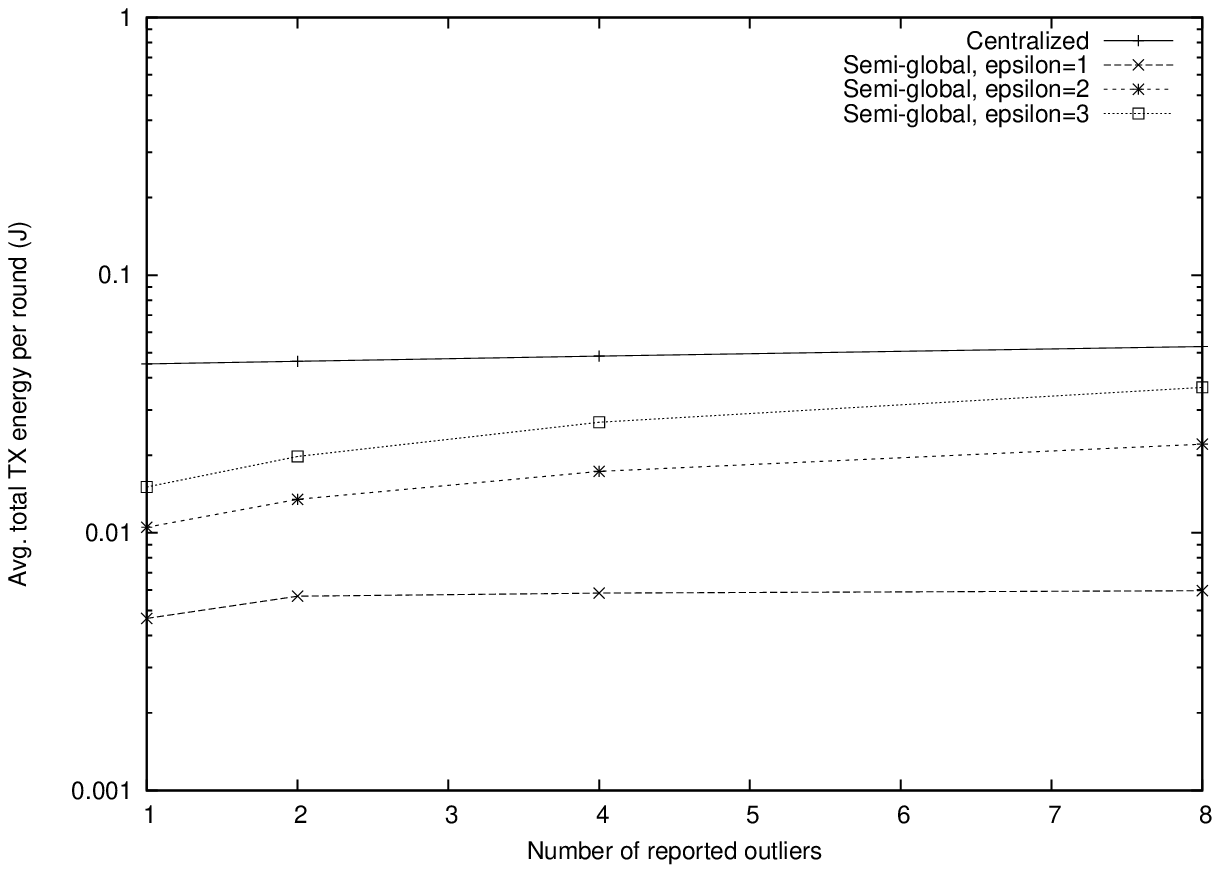}
\includegraphics[width=7cm]{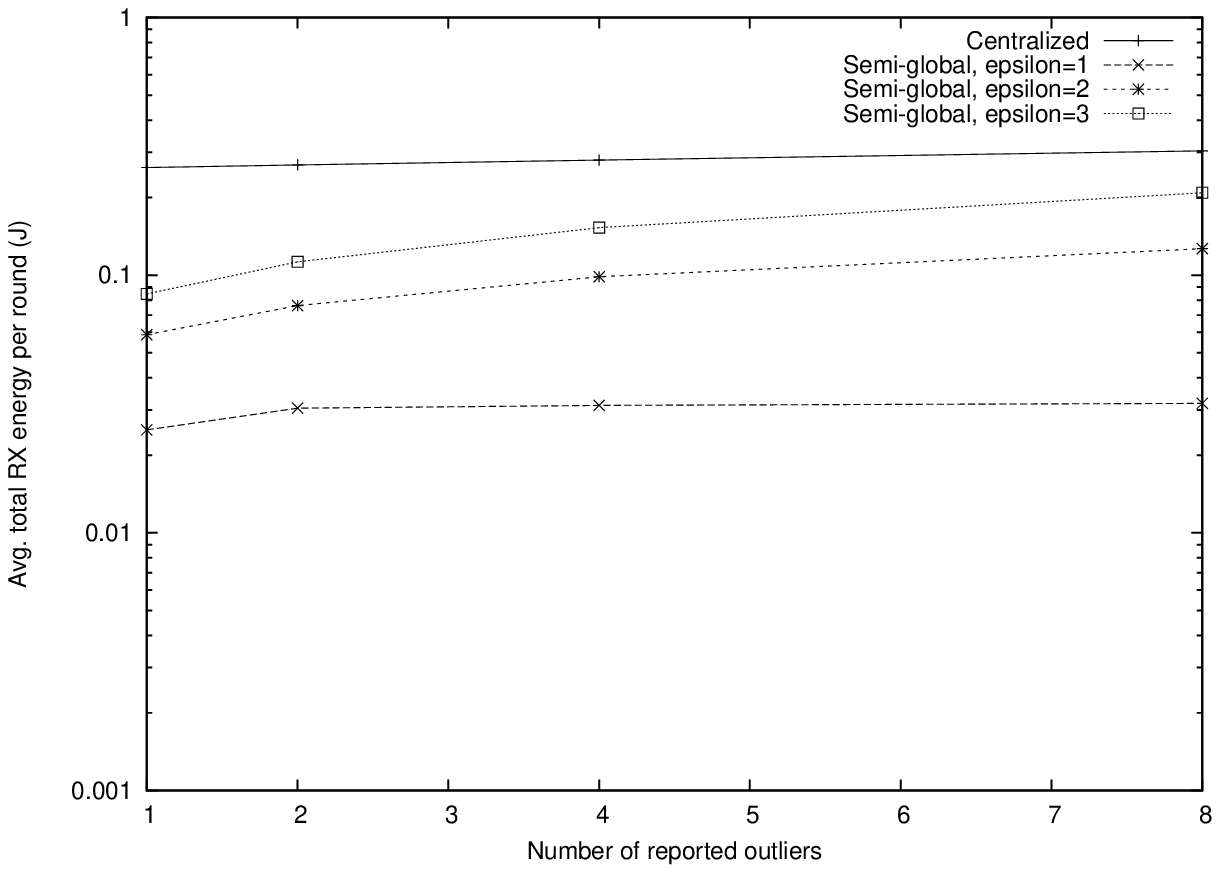}
\caption{Average transmission and receive energy consumed per node per sample interval vs. $n$ ($w$=20, $k$=4)
for localized outlier detection using k nearest neighbor outlier detection.}
\label{fig:local2_KNN}
\end{center}
\end{figure}

We now investigate how the number of outliers produced affects energy usage.
Figure \ref{fig:local2_KNN} shows the plots illustrating the performance of the localized
outlier detection algorithms under increasing values for $n$ for KNN
outlier detection. Similar plots for NN detection are omitted due to space restrictions
and similarity of results; NN detection is negligibly less energy efficient most likely due to a
lower rate of convergence. The energy usage
trends for these algorithms are straightforward and expected. Energy usage increases along
with both $n$ and $epsilon$, which both cause more message passing overhead with
increasing value. We also noticed that the rate at which energy usage increased
was related to $epsilon$. This is expected since the compounded effects of larger
$epsilon$ and $n$ values should make a more noticeable mark on how energy is used.


\section{Conclusions}

We addressed the problem of unsupervised outlier detection in WSNs. We developed a solution 
that

\begin{enumerate}
\item allows flexibility in the heuristic used to define outliers,

\item computes the result in-network to reduce both bandwidth and energy usage,

\item only uses single hop communication thus permitting very simple node failure detection
and message reliability assurance mechanisms (e.g., carrier-sense), and

\item  seamlessly accommodates dynamic updates to data.
\end{enumerate}

We evaluated the outlier detection algorithm's behavior on real-world sensor 
data using a simulated wireless sensor network. These initial results show 
promise for our algorithm in that it outperforms a strictly centralized 
approach under some very important circumstances. When the unabridged data 
from the entire sensor network are sent to a single location, the node 
collecting this data as well as its nearest neighbors become a bottleneck
of the entire system. Indeed, the density of traffic in this region is
proportional to the area of coverage of the entire network while the average
node has the traffic density proportional to the area covered by its 
communication range. In the example that we simulated in the paper, the
traffic in the area of the collecting node was about 50 times more dense than
in the other parts of the network. The immediate consequence is the shorter
life-time of the network, as the nodes near the collecting point will die
because of battery exhaustion when many remaining nodes will use just 2\% of 
their energy. The second consequence is the congestion of the traffic that 
either results in a lot of interference necessitating retransmissions or
delays or, alternatively, in delays imposed by a multi-slot bandwidth sharing 
scheme needed to avoid transmission interference. In short, using the 
centralized algorithm with its drastic imbalance of the traffic density 
will put even the best routing protocols under the sever stress. In contrast, 
our distributed and localized outlier detection algorithms avoid these
difficulties.

Our approach is well suited for applications in which the confidence of
an outlier rating may be calculated by either an adjustment of sliding 
window size or the number of neighbors used in a distance-based outlier 
detection technique. We assert that these applications are critical for 
resource-constrained sensor networks for two reasons. First, communication 
is a costly activity motivating the need for only the most accurate data 
to be transmitted to a client application. Second, emerging safety-critical 
applications that utilize wireless sensor networks will require the most 
accurate data, including outliers. This work represents
our contribution toward enabling efficient data cleaning solutions for these
types of applications.

\begin{acknowledgements}
The authors thank the U.S. National Science Foundation for support of Wolff, 
Giannella, and Kargupta through award IIS-0329143 and CAREER award 
IIS-0093353 and of Szymanski through award OISE-0334667. Research of Branch and Szymanski continued through participation in the
International Technology Alliance sponsored by the U.S. Army Research
Laboratory and the U.K. Ministry of Defense.
The authors thank Chris Morrell at Rensselaer Polytechnic Institute 
for his efforts in helping to obtain performance metrics and Wesley Griffin
at UMBC for his help in running simulations.
The authors also thank Samuel Madden at 
Massachusetts Institute of Technology and the team at the Intel Berkeley 
Research Lab for generating the sensor data used in this paper and assisting 
in its use. The content of this paper does
not necessarily reflect the position or policy of the U.S. Government or the MITRE Corporation ---no 
official endorsement should be inferred or implied.
C. Giannella completed this work primarily while in the Department of Computer Science, New Mexico State
University. 

\end{acknowledgements}

\bibliographystyle{spmpsci}      

\bibliography{Outliers_Journal}

\section{Appendix: Correctness Proofs for the Global Outlier Detection Algorithm} \label{appendix:proofs}

In this section, we provide detailed proofs of Theorems \ref{thm:termination} and \ref{thm:correctness}.
Before doing so, a few technical lemmas are needed.
The first two isolate a couple of
useful properties following from the axioms of $R$.

\begin{lem}
\label{lem:3}
For any $P \subseteq Q \subseteq D$ where $|P| \geq n$, 
if $O_{n}(P)$ $\neq$ $O_{n}(Q)$,
then there exists $x \in O_{n}(P)$ such that $R(x,P)$ $>$ 
$R(x,Q)$.
\end{lem}

\begin{proof}
Assume $O_{n}(P)$ $\neq$ $O_{n}(Q)$.
Since $|O_{n}(P)|$ $=$ $|O_{n}(Q)|$ $=$ $n$, then
there exists $x \in (O_{n}(P)\setminus O_{n}(Q))$ and
$y \in (O_{n}(Q)\setminus O_{n}(P))$.  Recall that we assume a tie-breaking
machanism is used to ensure $R(.P)$ and $R(.,Q)$ are one-to-one.  Thus,
by definition of $O_n(.)$ it follows that $R(x,P)$ $>$ $R(y,P)$ and $R(y,Q)$ $>$ $R(x,Q)$.  
The anti-monotoncity axiom implies $R(y,P)$ $\geq$ $R(y,Q)$ yeilding the desired result.
\qed
\end{proof}

\begin{lem}
\label{lem:1}
For any $P \subseteq D$, $x \in O_{n}(P)$, and $z \in P$, we have
$R(x,P)$ $=$ $R(x,[P|O_{n}(P)])$ $=$ $R(x,[P|O_n(P)] \cup \{z\})$.
\end{lem}

\begin{proof}
Since, by definition, $[P|x] \subseteq [P|O_n(P)] \subseteq 
([P|O_n(P)] \cup \{z\}) \subseteq P$, then by the
anti-monotonicity axiom it follows that

\begin{eqnarray*}
R(x,P) &=& R(x,[P|x]) \\
&\geq& R(x,[P|O_n(P)]) \\
&\geq& R(x,[P|O_n(P)] \cup \{z\}) \\
&\geq& R(x,P).
\end{eqnarray*} 
\qed
\end{proof}

The last technical lemma shows that once a sensor $p_i$ completes its
local computation, then $(D^i_{i,j} \cup D^i_{j,i})$
contains a particular crucial set of points (among others) needed for consistency among 
sensors' outlier estimates.

\begin{lem}
\label{lem:2}
For any $p_i$, once the main for-loop in the algorithm completes,  
$[P_i|O_{n}(D^i_{i,j} \cup D^i_{j,i})]$ 
$\subseteq$ $(D^i_{i,j} \cup D^i_{j,i})$.
\end{lem}

\begin{proof}
Let $D^i_{i,j}(before)$ denote the set of
points held by $p_i$ and sent from $p_i$ to $p_j$ immediately before the execution of the
``Repeat unil no change: ...'' step in the main for loop for $j \in \Gamma_i$.  For $\ell \geq 1$,
let $Z_j(\ell)$ denote $Z_j$ immediately before the $\ell^{th}$ iteration in the 
execution of the ``Repeat until no change ...'' step in the main for loop for $j \in \Gamma_i$, {\em e.g.}
$Z_j(1)$ $=$ $O_n(P_i)$ $\cup$ $[P_i|O_n(P_i)]$.  

By definition, $Z_j(\ell)$ $\subseteq$ $Z_j(\ell+1)$ $\subseteq$ $P_i$ and $P_i$ is finite.  Thus, let 
$\ell^*$ denote the smallest integer such that $Z_j(\ell^*-1)$ $=$ $Z_j(\ell^*)$.  Hence,
the ``Repeat until no change ...'' step terminates at the end of iteration $\ell^*$ and $Z_j = Z_j(\ell^*)$
in the remainder of the main for loop.  Therefore,

$$Z_j(\ell^*) = Z_j(\ell^*) \cup [P_i|O_n(D^i_{i,j}(before) \cup D^i_{j,i} \cup Z_j(\ell^*))]$$

\noindent and

$$D^i_{i,j} \cup D^i_{j,i} = D^i_{i,j}(before) \cup D^i_{j,i} \cup Z_j(\ell^*).$$

\noindent It follows that $[P_i|O_n(D^i_{i,j} \cup D^i_{j,i})] \subseteq Z_j(\ell^*) \subseteq (D^i_{i,j} \cup D^i_{j,i})$.  
\qed
\end{proof}

\noindent Now we prove that upon termination of the algorithm, the sensors' estiamtes are 
consistent.

\textit{Theorem}\ref{thm:termination}
Assuming a connected network, if for all sensors $p_i$: $D_i$ and $\Gamma_i$ do not change, then upon termination of 
the algorithm all sensors' outlier estimates and supports agree: for all $p_i, p_j$: 
(i) $O_n(P_i) = O_n(P_j)$ and (ii) $[P_i|O_n(P_i)] = [P_j|O_n(P_j)]$.

\begin{proof} 
Since the network is connected, we may assume, without loss of generality, that $p_i$ and $p_j$
are neighbors.  
To prove part (i), we will show that $O_{n}(P_i)$ $=$ 
$O_{n}(D^i_{i,j} \cup D^i_{j,i})$ $=$
$O_{n}(D^j_{i,j} \cup D^j_{j,i})$ $=$
$O_{n}(P_j)$.  The middle equality follows from the fact that
$(D^i_{i,j} \cup D^i_{j,i})$ $=$
$(D^j_{i,j} \cup D^j_{j,i})$.  By symmetry,
it suffices to show the first equality.  

Suppose 
$O_{n}(P_i)$ $\neq$ 
$O_{n}(D^i_{i,j} \cup D^i_{j,i})$.  The following 
contradiction is reached.  There exists
$x \in O_{n}(D^i_{i,j} \cup D^i_{j,i})$ such that

\begin{eqnarray*}
R(x,D^i_{i,j} \cup D^i_{j,i}) &>& R(x,P_i) \\
&=&  R(x,[P_i|x]) \\
&\geq& R(x,[P_i|O_{n}(D^i_{i,j} \cup D^i_{j,i})]) \\
&\geq& R(x,D^i_{i,j} \cup D^i_{j,i}).
\end{eqnarray*}

\noindent The first inequality follows from Lemma \ref{lem:3} (with
$P$ $=$ $(D^i_{i,j} \cup D^i_{j,i})$ and $Q$ $=$ 
$P_i$).  The equality follows from the definition of support $[.|.]$.  The last two
inequalities follow from the anti-monotonicity of $R$
and Lemma \ref{lem:2}.  

To prove part (ii) $[P_i|O_n(P_i)] = [P_j|O_n(P_j)]$,
it suffices to show that for any $x \in$ 
$O_{n}(D^i_{i,j} \cup D^i_{j,i})$, it is the
case that $[P_i|x]$ $=$ $[P_j|x]$. This is because $O_{n}(P_i)$ $=$
$O_{n}(D^i_{i,j} \cup D^i_{j,i})$ $=$ 
$O_{n}(P_j)$.  We will prove that $[P_i|x]$ $=$ $[P_i \cap P_j|x]$ $=$ $[P_j|x]$.
By symmetry it is enough to show the first equality.

From Lemma \ref{lem:2} it follows that 

\begin{eqnarray*}
[P_i|x] &\subseteq& [P_i|O_{n}(D^i_{i,j} \cup D^i_{j,i})] \\
&\subseteq& (D^i_{i,j} \cup D^i_{j,i}) \\
&\subseteq& (P_i \cap P_j).
\end{eqnarray*}

\noindent Thus, anti-monotonicity implies
$R(x,P_i)$ $\geq$ $R(x,[P_i|x])$ $\geq$ $R(x,P_i \cap P_j)$ $\geq$ $R(x,P_i)$, and so,

\begin{eqnarray*}
R(x,P_i) = R(x,[P_i|x]) = R(x, P_i \cap P_j).
\end{eqnarray*}

Therefore, $[P_i|x]$, $[P_i \cap P_j|x]$ are support sets of $x$ with respect to $P_i \cap P_j$ and $P_i$.
Since $[P_i|x]$ ($[P_i \cap P_j|x]$) is the {\em unique} smallest support set for $x$ with respect to
$P_i$ ($P_i \cap P_j$), then it follows that $[P_i|x]$ $=$ $[P_i \cap P_j|x]$.  
\qed
\end{proof}

\noindent Finally, we prove that upon termination the sensors' estiamtes
are equal to the correct answer.

\textit{Theorem}\ref{thm:correctness}
Assuming a connected network, if for all sensors $p_i$: $D_i$ and $\Gamma_i$ do not change, then upon termination of the 
algorithm, all sensors' outlier estimate will be correct: for all $p_i$: 
$O_n(P_i) = O_n(D)$.

\begin{proof}
Suppose there exists a sensor $p_i$ such that $O_{n}(P_i)$ $\neq$ 
$O_{n}(D)$.  By Lemma \ref{lem:3} (with $P$ $=$ $P_i$ and $Q$ $=$ $D$), there
exists $x \in O_{n}(P_i)$ such that $R(x,P_i)$ $>$ $R(x,D)$.
Moreover, the first equality in Lemma 
\ref{lem:1} (with $P = P_i$) implies that 
$R(x,[P_i|O_{n}(P_i)])$ $=$ $R(P_i,x)$.

Since $R(x,[P_i|O_{n}(P_i)])$ $>$ $R(x,D)$, then the smoothness axiom 
(with $Q_1 = [P_i|O_n(P_i)]$ and
$Q_2 = D$), 
implies there exists $z \in (D \setminus [P_i|O_{n}(P_i)])$ 

\begin{eqnarray*}
R(x,[P_i|O_{n}(P_i)]) > R(x,[P_i|O_{n}(P_i)]) \cup \{z\}).
\end{eqnarray*}

\indent This point $z$ must be contained in $P_j$ for some sensor $p_j$.
Hence, the inequality the following contradiction is reached.

\begin{eqnarray*}
R(x,[P_i|O_{n}(P_i)]) &>& R(x,[P_i|O_{n}(P_i)] \cup \{z\}) \\
&=& R(x,[P_i|O_{n}(P_j)] \cup \{z\}) \\
&=& R(x,[P_j|O_{n}(P_j)]) \\
&=& R(x,[P_i|O_{n}(P_i)]).
\end{eqnarray*}

\noindent The inequality above leads to the following contradiction.  The first
equality follows from Theorem \ref{thm:termination} part (i).  The 
middle equality follows from the second equality of Lemma \ref{lem:1} 
(with $P = P_j$ and noting that $O_{n}(P_j) = O_{n}(P_i)$ by Theorem 
\ref{thm:termination} part (i)).  The last equality follows from Theorem
\ref{thm:termination} part (ii).    
\qed
\end{proof}

\end{document}